\documentclass[11pt,letterpaper]{article}

\usepackage[margin=1in]{geometry}                                                    
\usepackage{amsmath,amsfonts,amssymb,color}
\usepackage{tikz}
\usepackage{pmat}
\usepackage{amsthm}
\usepackage{empheq}

\newcommand*\widefbox[1]{\fbox{\hspace{0.2em}#1\hspace{0.2em}}}
\newcommand{\mycaption}[1]{\stepcounter{figure}\raisebox{-7pt}
  {\footnotesize Fig. \thefigure.\hspace{3pt} #1}}
\newtheorem{problem}{Problem}
\newtheorem{condition}{Condition}
\newtheorem{definition}{Definition}
\newtheorem{theorem}{Theorem}
\newtheorem{lemma}{Lemma}
\newtheorem{proposition}{Proposition}

\newtheorem{remark}{Remark}
\newtheorem{assumption}{Assumption}

\usepackage{epsfig}

\newcommand*{\Scale}[2][4]{\scalebox{#1}{$#2$}}

\usepackage{algorithm}
\usepackage{algorithmic}
\usepackage{epsfig}

\usepackage{array}
\usepackage{multirow}
\usepackage{epstopdf}
\usepackage{tikz}
\usepackage{relsize}
\usepackage[pdfpagelabels,pdfpagemode=None,breaklinks=true]{hyperref}\hypersetup{
colorlinks = true,
citecolor = blue}
\usetikzlibrary{shapes,snakes}

\usepackage{cite}
\newcommand*{\TitleFont}{%
\usefont{\encodingdefault}{\rmdefault}{}{n}%
\fontsize{16}{20}%
\selectfont}

\DeclareMathAlphabet{\mathcal}{OMS}{cmsy}{m}{n} 

\title{\TitleFont Distributed Observers for LTI Systems}
\author{Aritra~Mitra~and~Shreyas Sundaram
\thanks{The authors are with the School of Electrical and Computer Engineering at Purdue University. Email: {\tt \{mitra14,sundara2\}@purdue.edu}}%
}

\begin{document}
\maketitle

\begin{abstract}
We consider the problem of distributed state estimation of a linear time-invariant (LTI) system by a network of sensors. We develop a distributed observer that guarantees asymptotic reconstruction of the state for the most general class of LTI systems, sensor network topologies and sensor measurement structures. Our analysis builds upon the following key observation - a given node can reconstruct a portion of the state solely by using its own measurements and constructing appropriate Luenberger observers; hence it only needs to exchange information with neighbors (via consensus dynamics) for estimating the portion of the state that is not locally detectable. This intuitive approach leads to a new class of distributed observers with several appealing features. Furthermore, by imposing additional constraints on the system dynamics and network topology, we show that it is possible to construct a simpler version of the proposed distributed observer that achieves the same objective while admitting a fully distributed design phase. Our general framework allows extensions to time-varying networks that result from communication losses, and scenarios including faults or attacks at the nodes.
\end{abstract}

\section{Introduction}
In many applications involving large-scale complex systems (such as the power grid, transportation systems, industrial plants, etc.), the state of the system is monitored by a group of sensors spatially distributed over large sparse networks where the communication between sensors is limited (see \cite{survey1,survey2}). To model such a scenario, consider the discrete-time linear time-invariant dynamical system
\begin{equation}
\mathbf{x}[k+1] = \mathbf{Ax}[k],
\label{eqn:plant}
\end{equation}
where $k \in \mathbb{N}$ is the discrete-time index, $\mathbf{x}[k] \in {\mathbb{R}}^n$ is the state vector and  $\mathbf{A} \in {\mathbb{R}}^{ n \times n} $ is the system matrix. The state of the system is monitored by a network\footnote{The terms `network' and `communication graph' are used interchangeably throughout the paper.} of $N$ sensors, each of which receives a partial measurement of the state at every time-step. Specifically, the $i$-th sensor has access to a measurement of the state, given by
\begin{equation}
\mathbf{y}_{i}[k]=\mathbf{C}_i\mathbf{x}[k],
\label{eqn:Obsmodel}
\end{equation}
where $\mathbf{y}_{i}[k] \in {\mathbb{R}}^{r_i}$ and $\mathbf{C}_i \in {\mathbb{R}}^{r_i \times n}$. We use $\mathbf{y}[k]={\begin{bmatrix}\mathbf{y}^T_{1}[k] & \cdots & \mathbf{y}^T_{N}[k]\end{bmatrix}}^T$ to represent the collective measurement vector, and $\mathbf{C}={\begin{bmatrix}\mathbf{C}^T_{1} & \cdots & \mathbf{C}^T_{N}\end{bmatrix}}^T$ to denote the collection of the sensor observation matrices. These sensors are represented as nodes of an underlying directed communication graph which governs the information flow between the sensors.

Each node is capable of exchanging information with its neighbors and performing computational tasks. The goal of each node is to estimate the entire system state $\mathbf{x}[k]$ based on its respective (limited) state measurements and the information obtained from neighbors. This is known as the \textit{distributed state estimation problem}.

In this paper, our \textbf{objective} is to design a distributed algorithm that guarantees asymptotic reconstruction of the entire state $\mathbf{x}[k]$ at each sensor node.\footnote{The problem is formally stated in Section \ref{sec:prob_form}.} For much of the paper, we will focus on  developing theory for \textit{linear time-invariant systems} and \textit{time-invariant directed communication graphs.} In Section \ref{sec:timevar}, however, we shall establish that our proposed framework can be extended to account for certain types of time-varying networks that may arise as a consequence of intermittent communication link failures.

The paper is organized as follows. In Section \ref{sec:sysmod}, we formally describe the problem, discuss related work and summarize our contributions. Some preliminary ideas and terminology required for subsequent analysis are presented in Section \ref{sec:Prelim}. Section \ref{illustration} highlights the key ideas of our distributed estimation scheme via a simple illustrative example. In Section \ref{Estimation}, we solve the most general version of the problem whereas in Section \ref{sec:Case1} we provide a solution strategy for a simpler variant of the original problem that enjoys several implementation benefits. We discuss the extension of our framework to time-varying networks in Section \ref{sec:timevar} and provide a simulation example in Section \ref{sec:example}.  Conclusions and avenues for future work are presented in Section \ref{sec:conclusion}.

\section{System Model}
\label{sec:sysmod}
\subsection{Notation}
\label{sec:notation}
A directed graph is denoted by $\mathcal{G} =(\mathcal{V},\mathcal{E})$, where $\mathcal{V} =\{1, \cdots, N\}$ is the set of nodes and $\mathcal{E} \subseteq \mathcal{V} \times \mathcal{V} $ represents the edges. An edge from node $j$ to node $i$, denoted by $(j,i)$, implies that node $j$ can transmit information to node $i$. The neighborhood of the $i$-th node is defined as $\mathcal{N}_i \triangleq \{i\} \cup \{j\,|\,(j,i) \in \mathcal{E} \}.$  The notation $|\mathcal{V}|$ is used to denote the cardinality of a set $\mathcal{V}$. Throughout the rest of this paper, we use the terms `nodes' and `sensors' interchangeably.

The set of all eigenvalues of a matrix $\mathbf{A}$ is denoted by $sp(\mathbf{A}) \triangleq \{\lambda \in \mathbb{C}\,|\,det(\mathbf{A}-\lambda\mathbf{I}) = 0\}$. The set of all marginally stable and unstable eigenvalues of a matrix $\mathbf{A}$ is denoted by $\Lambda_{U}(\mathbf{A}) \triangleq \{\lambda \in sp(\mathbf{A})\,|\, |\lambda| \geq 1 \}$. For a matrix $\mathbf{A}$, we use $a_\mathbf{A}(\lambda)$ and $g_\mathbf{A}(\lambda)$ to denote the algebraic and geometric multiplicities, respectively, of an eigenvalue $\lambda \in sp(\mathbf{A})$. An eigenvalue $\lambda$ is said to be simple if $a_\mathbf{A}(\lambda)=g_\mathbf{A}(\lambda)=1$. For a set $\{\mathbf{A}_1, \cdots, \mathbf{A}_n\}$ of matrices, we use the notation $diag(\mathbf{A}_1, \cdots, \mathbf{A}_n)$ to refer to a block diagonal matrix with the matrix $\mathbf{A}_i$ as the $i$-th block-diagonal entry.
For a set $\mathcal{S}=\{s_1,\cdots, s_p\} \subseteq \{1,\cdots, N\}$, and a matrix $\mathbf{C}={\begin{bmatrix} \mathbf{C}^T_{1} & \cdots &  \mathbf{C}^T_{N}\end{bmatrix} }^T$, we define $\mathbf{C}_{\mathcal{S}}\triangleq {\begin{bmatrix}\mathbf{C}^{T}_{s_1} & \cdots & \mathbf{C}^{T}_{s_p}\end{bmatrix}}^T$. We use the star notation to avoid writing matrices that are either unimportant or that can be inferred from context. We use $\mathbf{I}_{r}$ to indicate an identity matrix of dimension $r \times r$.

\subsection{Problem Formulation}
\label{sec:prob_form}
Consider the LTI system given by \eqref{eqn:plant}, the measurement model specified by \eqref{eqn:Obsmodel}, and a predefined directed communication graph $\mathcal{G=(V,E)}$, where $\mathcal{V}$ represents the set of $N$ nodes (or sensors). Each node $i$ maintains an estimate $\hat{\mathbf{x}}_i[k]$ of the state $\mathbf{x}[k]$ of system \eqref{eqn:plant}, and updates such an estimate based on information received from its neighbors and its local measurements (if any). To formally define the problem under study, we use the following terminology.

\begin{definition} [Distributed Observer]
\label{defn:distobs}
A set of state estimate update and information exchange rules is called a distributed observer if $\lim_{k\to\infty} ||\hat{\mathbf{x}}_i[k]-\mathbf{x}[k]||=0, \forall i \in \{1, \cdots, N\}$, i.e., the state estimate maintained by each node asymptotically converges to the true state of the plant. 
\end{definition}

There are various technical challenges associated with constructing a distributed observer. First, if the pair $(\mathbf{A}, \mathbf{C}_i)$ is not detectable for some (or all) $i \in \{1, \cdots, N\}$, then the corresponding nodes cannot estimate the true state of the plant based on their own local measurements, thereby dictating the need to exchange information with other nodes. Second, this exchange of information is restricted by the underlying communication graph $\mathcal{G}$. With these challenges in mind, we address and solve the following problem in this paper.

\begin{problem} Design a distributed observer for LTI systems of the form \eqref{eqn:plant}, linear measurement models of the form \eqref{eqn:Obsmodel}, and time-invariant directed communication graphs.
\end{problem}

There are a variety of approaches to construct distributed observers (as defined in Definition \ref{defn:distobs}) that have been proposed in the literature, which we will now review. After that, we will summarize how our approach differs from the existing approaches, before delving into the details of our construction.

\subsection{Related Work}
\label{sec:relwork}
The papers \cite{scalar1,scalar2,scalar3} consider distributed estimation of scalar stochastic dynamical systems over general graphs; in these works, it is typically assumed that each node receives scalar local observations, leading to local observability at every node. The papers \cite{decent1,decent2} considered a version of this problem where the underlying communication graph is assumed to be complete. For more general stochastic systems, the Kalman filtering based approach to solving the distributed estimation problem has been explored by several researchers. The approach proposed in \cite{olfati1,olfati2,olfati3} relies on a two-step strategy - a Kalman filter based state estimate update rule, and a data fusion step based on average-consensus. The stability and performance issues of this method have been investigated in \cite{kalman1,kalman2}. A drawback of this method (and the ones in \cite{scalar2,infinite1,infinite2}), stems from the fact that they require a (theoretically) infinite number of data fusion iterations between two consecutive time steps of the plant dynamics in order to reach average consensus, thereby leading to a two-time-scale algorithm. More recently, finite-time data fusion has been studied in \cite{Khan1} and \cite{ICF}. Although an improvement over the infinite-time data fusion case, these methods still rely on a two-time-scale strategy. In \cite{Baras}, the authors employ an LMI-based approach for obtaining sub-optimal filter gains that minimize a quadratic cost function of the covariance matrices of the estimation errors.
  
In \cite{Khanobs1} and \cite{Khanobs2}, the authors consider the distributed observer design problem for undirected graphs; in these works, they propose a \textit{scalar-gain estimator} that runs on a single-time-scale.\footnote{By a single-time-scale algorithm, we imply an algorithm where each node operates at the same time-scale as the plant, and updates its estimate and transmits information to neighbors only once in each time-step.} They introduce the notion of ``Network Tracking Capacity" (NTC), a measure of the most unstable dynamics (in terms of the 2-norm of the state matrix) that can be estimated with bounded mean-squared error under their scheme. However, the tight coupling between the network and the plant dynamics typically limits the
set of unstable eigenvalues that can be accommodated by their method without violating the constraints imposed upon the range of the scalar gain parameter. In \cite{ugrinov}, the author approaches the observer design problem from a geometric perspective and provides separate necessary and sufficient conditions for consensus-based distributed observer design. In \cite{martins1,martins2,martins3,wang}, the authors use single-time-scale algorithms, and work under the broadest assumptions, namely that the pair $(\mathbf{A,C})$ is detectable, where $\mathbf{A}$ represents the system matrix, and $\mathbf{C}$ is the collection of all the node observation matrices. In all of these works, the authors rely on state augmentation\footnote{In these works, some sensor nodes maintain observers of dimension larger than that of the state of the plant; hence, such observers are referred to as augmented observers, and the state they estimate is referred to as an augmented state.} for casting the distributed estimation problem as a problem of designing a decentralized stabilizing controller for an LTI plant, using the notion of fixed modes \cite{fixed1,fixed2}. Specifically, in \cite{martins3}, the authors relate the distributed observer design problem for directed networks to the detectability of certain strongly connected clusters within the network, and provide a single necessary and sufficient condition for their scheme.

\subsection{Summary of Contributions}
In this paper, we provide a new approach to designing distributed observers for LTI dynamical systems. Specifically, we use the following simple, yet key observation - for each node, there may be certain portions of the state that the node can reconstruct using only its local measurements. The node thus does so. For the remaining portion of the state space, the node relies on a consensus-based update rule. The key is that those nodes that can reconstruct certain states on their own act as ``\textit{root nodes}" (or ``\textit{leaders}") in the consensus dynamics, leading the rest of the nodes to asymptotically estimate those states as well. These ideas, in a nutshell, constitute the essence of our distributed estimation strategy.

We begin by considering the most general category of systems and graphs (taken together) for which a distributed observer can be constructed, and develop an estimation scheme that enjoys the following appealing features \textit{simultaneously}, thereby differentiating our work from the existing literature discussed in Section \ref{sec:relwork}: i) it provides theoretical guarantees regarding the design of asymptotically stable estimators; (ii) it results in a single-time-scale algorithm; (iii) it does not require any state augmentation; (iv) it requires only state estimates to be exchanged locally; and (v) it works under the broadest conditions on the system and communication graph. Subsequently, for a certain subclass of systems and communication graphs, we provide a simpler fully distributed estimation scheme (at both design-  and run-time) for achieving asymptotic state reconstruction. Finally, we show that our proposed framework can be extended to guarantee asymptotic state reconstruction in the presence of communication losses that lead to time-varying networks.

Some of the results from Section \ref{Estimation} of the paper appeared (without proofs) in \cite{allerton}. This journal paper substantially expands upon the conference paper by providing full proofs of all results, a detailed analysis of classes of systems and graphs that allow efficient distributed implementations (at both the design- and run-time phases), an analysis of the robustness of our general framework to communication losses, and examples and simulations to complement the proposed theory.

\section{Preliminaries}
\label{sec:Prelim}
Before we proceed with a formal analysis of the problem of designing a distributed observer, we first identify the main consideration that shall dictate our solution strategy, namely, \textit{the relationship between the measurement structure of the nodes and the underlying communication graph}. To classify sets of systems and graphs based on this relationship, we need to first establish some notation. Accordingly, for each node $i$, we denote the detectable and undetectable eigenvalues\footnote{Given a pair $(\mathbf{A},\mathbf{C}_i)$, an eigenvalue $\lambda \in \Lambda_{U}(\mathbf{A})$ is said to be detectable if $rank\left[\begin{smallmatrix} \mathbf{A}-\lambda\mathbf{I}_n \\ \mathbf{C}_i \end{smallmatrix}\right] = n$. Each stable eigenvalue of $\mathbf{A}$ is by default considered to be detectable.} of $\mathbf{A}$ by the sets $\mathcal{O}_{i}$ and $\mathcal{UO}_{i}$, respectively. We define $\sigma_i \triangleq |\mathcal{O}_{i}|$. Next, we introduce the notion of \textit{root nodes}.

\begin{definition}[Root nodes]
\label{defn:rootnodes}
For each $\lambda_j \in \Lambda_{U}(\mathbf{A})$, the set of nodes that can detect $\lambda_j$ is denoted by $\mathcal{S}_j$, and called the set of \textit{root nodes} for $\lambda_j$.\footnote{Throughout the paper, for the sake of conciseness, we use the terminology `node $i$ can detect eigenvalue $\lambda_j$' to imply that $rank\left[\begin{smallmatrix} \mathbf{A}-\lambda_j\mathbf{I}_n \\ \mathbf{C}_i \end{smallmatrix}\right] = n$.}
\end{definition}

We also recall the definition of a source component of a graph \cite{martins3}.

 \begin{definition} [Source Component] Given a directed graph $\mathcal{G} =(\mathcal{V},\mathcal{E})$, a source component $(\mathcal{V}_s,\mathcal{E}_s)$ is defined as a strongly connected component of $\mathcal{G}$ such that there are no edges from $\mathcal{V} \setminus \mathcal{V}_s$ to $\mathcal{V}_s$.
\end{definition}

Let there be $p$ source components of $\mathcal{G}$, denoted by ${\{(\mathcal{V}_i,\mathcal{E}_i)\}}_{i \in \{1, \cdots , p\}}$. The subsystem associated with the $i$-th source component is given by the pair $(\mathbf{A},\mathbf{C}_{\mathcal{V}_i})$. For the subsequent development, it should be noted that by a system $(\mathbf{A,C})$, we refer to the matrix $\mathbf{A}$ in equation  (\ref{eqn:plant}), and the matrix $\mathbf{C}={\begin{bmatrix}\mathbf{C}^T_{1} & \cdots & \mathbf{C}^T_{N}\end{bmatrix}}^T$ containing each of the measurement matrices given by (\ref{eqn:Obsmodel}). Then, we classify systems and graphs based on the following two conditions.

\begin{condition}
A system $(\mathbf{A,C})$ and graph $\mathcal{G}$ are said to satisfy Condition 1 if the sub-system associated with every source component is detectable, i.e., the pair $(\mathbf{A},\mathbf{C}_{\mathcal{V}_i})$ is detectable $\forall i  \in \{1, \cdots , p\}$.
\end{condition}

\begin{condition}
A system $(\mathbf{A,C})$ and graph $\mathcal{G}$ are said to satisfy Condition 2 if for each unstable or marginally stable eigenvalue of the plant, there exists at least one root node within each source component, i.e., for all $i \in \{1, \cdots , p\}$ and all $\lambda_j \in \Lambda_{U}(\mathbf{A})$, there exists $l \in \mathcal{V}_i$, such that $rank\begin{bmatrix} \mathbf{A}-\lambda_j\mathbf{I}_n \\ \mathbf{C}_l \end{bmatrix} = n$.\footnote{Note that given a source component, Condition $2$ does not necessarily imply the existence of a single node within such a component that can simultaneously detect all the unstable and marginally stable eigenvalues of the system via its own measurements.}
\end{condition}

\begin{remark}
\label{remark:Cond}
It is trivial to see that if a system $(\mathbf{A,C})$ and graph $\mathcal{G}$ satisfy Condition 2, they also satisfy Condition 1. To see that the converse is not true in general, 
 consider the $3$-node network $\mathcal{G}$ in Figure \ref{fig:example}, and the following model:
\begin{equation}
\mathbf{A}=\begin{bmatrix} 2 & 0\\ 0 & 2 \end{bmatrix}, \mathbf{C}_1= \begin{bmatrix} 1 & 0 \end{bmatrix}, \mathbf{C}_2= \begin{bmatrix} 0 & 1 \end{bmatrix}, \mathbf{C}_3 = \begin{bmatrix} 1 & 0\\ 0 & 1 \end{bmatrix}.
\label{eqn:examplesystem}
\end{equation}
From Figure \ref{fig:example}, we see that the network has two source components, namely,  the strong component formed by nodes $1$ and $2$ ($S_1$), and the isolated node $3$ ($S_2$). Clearly, each of the pairs $(\mathbf{A,C}_3)$ and $(\mathbf{A},\begin{bmatrix} \mathbf{C}_1 \\ \mathbf{C}_2 \end{bmatrix})$ are detectable. Thus, the system is detectable from each of the two source components. It follows that this system and graph satisfy Condition 1. However, neither node $1$ nor node $2$ can detect the eigenvalue $\lambda=2$ based on just their own measurements, i.e., there does not exist a root node for $\lambda=2$ within source component $S_1$. Thus, this system and graph do not satisfy Condition 2.
\end{remark}

\begin{figure}[t]
\begin{center}
\begin{tikzpicture}
[->,shorten >=1pt,scale=.5, minimum size=10pt, auto=center, node distance=3cm,
  thick, every node/.style={circle, draw=black, thick},]
\tikzstyle{block} = [rectangle, draw, fill=green!20, 
    text width=5em, text centered, rounded corners, minimum height=4em];
    \node [block]  at (3,4) (plant) {LTI plant};
\node [circle, draw, fill=blue!20](n1)  at (0,0)  (1) {1};
\node [circle, draw, fill=blue!20](n2) at (3,0)   (2)  {2};
\node [circle, draw, fill=blue!20](n3) at (6,-0)  (3)  {3};

\path[every node/.style={font=\sffamily\small}]
(plant)
             edge[red, bend right] node [right] {$y_1[k]$} (1)
             edge[red] node [right] {$y_2[k]$} (2)
             edge[red, bend left] node [right] {$y_3[k]$} (3)
                       
    (1)
        edge [bend right] node [] {} (2)
        
    (2)
        edge [bend right] node [] {} (1);

\end{tikzpicture}
\end{center}
\caption{Example for illustrating Remark \ref{remark:Cond}.}
\label{fig:example}
\end{figure}
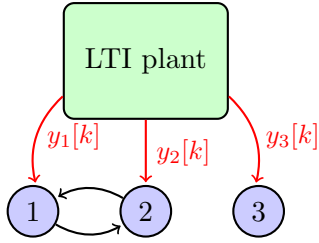

In \cite{martins3}, the authors identified that a distributed observer cannot be constructed (regardless of the state update or exchange rules) if the system $(\mathbf{A,C})$ and graph $\mathcal{G}$ do not satisfy Condition 1. They then designed a distributed observer for the class of systems and graphs satisfying Condition 1 by constructing augmented state observers (i.e., observers of dimension larger than that of the system) drawing upon connections to decentralized control theory. Here, we present an alternate and more direct design approach, and in the process, establish that it is possible to design a distributed observer \textit{without state augmentation} for this (most general) class of systems and graphs.\footnote{The exact structure of our distributed observer  presented in Section \ref{sec:compact} illustrates that the dimension of the internal state/estimate $\hat{\mathbf{x}}_i[k]$ maintained by a given node $i$ is equal to the dimension of the state $\mathbf{x}[k]$.} Before we delve into the specifics of the distributed observer design for systems and graphs satisfying Condition 1, we present a simple motivating example which serves to build intuition for the more complicated scenarios.\footnote{At this point, it is worth mentioning that although the distributed observer that we shall design for systems and graphs satisfying Condition 1 will also work for systems and graphs satisfying Condition 2, we will later propose an alternate scheme with various implementation benefits for the latter class of systems and graphs.}

\section{Illustrative Example}
\label{illustration}
\begin{figure}[t]
\begin{center}
\begin{tikzpicture}
[->,shorten >=1pt,scale=.5, minimum size=10pt, auto=center, node distance=3cm,
  thick, every node/.style={circle, draw=black, thick},]
\tikzstyle{block} = [rectangle, draw, fill=green!20, 
    text width=5em, text centered, rounded corners, minimum height=4em];
\node [block]  at (0,3.5) (plant) {LTI plant};

\node [circle, draw, fill=red!50](n1) at (0,0)  (1)  {1};
\node [circle, draw, fill=blue!20](n2) at (-2,-2)   (2)  {2};
\node [circle, draw, fill=blue!20](n3) at (2,-2)   (3)  {3};

\path[every node/.style={font=\sffamily\small}]
(plant)
             edge[red] node [right] {$y_1[k]$} (1)
             
    (1)
        edge [bend right] node [] {} (2)
        edge [bend left]  node [] {} (3)
    (2)
        edge [bend right] node [] {} (1);
        
\node [block]  at (8,3.5) (plant2) {LTI plant};        
        
\node [circle, draw, fill=red!50](n1) at (8,0)     (4)  {1};
\node [circle, draw, fill=blue!20](n2) at (6,-2)   (5)  {2};
\node [circle, draw, fill=blue!20](n3) at (10,-2)   (6)  {3};

\path[every node/.style={font=\sffamily\small}]
(plant2)
             edge[red] node [right] {$y_1[k]$} (4)

    (4)
        edge [bend right] node [] {} (5)
        edge [bend left]  node [] {} (6);
       
\end{tikzpicture}
\end{center}
\caption{The graph on the left is the actual network. The graph on the right is a DAG constructed from the original graph.}
\label{fig:DAG}
\end{figure}
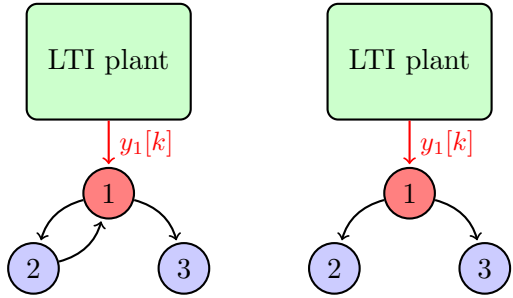

Consider a scalar unstable plant with dynamics given by $x[k+1]=1.5x[k]$. The plant is monitored by a network of nodes, as depicted by Figure \ref{fig:DAG}. Node $1$ has a measurement given by $y_1[k]=x[k]$, whereas nodes $2$ and $3$ have no measurements. Given this plant and network model, we wish to design a distributed observer. The commonly adopted approach in the literature is to develop a consensus-based state estimate update rule for each node in the network \cite{martins1,martins2,martins3,Khanobs1,Khanobs2}. Here, we make the following observation: since node 1 can detect the eigenvalue $\lambda=1.5$ of the plant based on its own measurements, it can run a Luenberger observer for estimating $x[k]$, without requiring information from its neighbors. Specifically, the following Luenberger observer allows node $1$ to estimate and predict the state:
\begin{equation}
\hat{x}_{1}[k+1]=1.5\hat{x}_{1}[k]+1.5(y_1[k]-\hat{x}_{1}[k])=1.5y_1[k].
\label{eqn:motivatingex}
\end{equation}
Here, $\hat{x}_{1}[k]$ is the estimate of $x[k]$ maintained by node $1$ at time-step $k$. Now suppose  nodes $2$ and $3$ update their respective estimates of $x[k]$ as follows: $\hat{x}_{i}[k+1]=1.5\hat{x}_{1}[k], i=2,3$. Since $\lim_{k\to\infty}|\hat{x}_{1}[k]-x[k]|=0$ based on the Luenberger observer dynamics given by (\ref{eqn:motivatingex}), it is easy to see that the estimates of nodes 2 and 3 also converge to the true state $x[k]$. This simple example illustrates the following key observations. (i) It is not necessary for every node in the network to run consensus dynamics for estimating the state. More generally, a node needs to run consensus for estimating only the portion of the state vector that is not locally detectable. The rest of the state can be estimated via appropriately designed Luenberger observers. (ii) An inspection of the observable subspace of each node guides the decision of participating (or not participating) in consensus for the example we considered. For more general system and measurement matrices, we shall rely on appropriate similarity transformations which shall reveal what a node can or cannot observe. (iii) Although node $1$ is in a position to receive information from node $2$, it chooses not to listen to any of its neighbors. This pattern of information flow results in a special Directed Acyclic Graph (DAG) of the original network, rooted at node 1. In the DAG constructed in the illustrative example, node $1$ can be viewed as the \textit{source} of information for the state $x[k]$, and the DAG structure can be viewed as the medium for transmitting information from the source to the rest of the network, without corrupting the source itself (this is achieved in this example by ignoring the edge from node $2$ to node $1$). Under this approach, note that every node maintains an observer of dimension $1$, which is equal to the dimension of the state (i.e., there is no state augmentation). Based on these observations, we are now ready to extend the ideas conveyed by this simple example for tackling more general systems and networks.\footnote{Notice that the original network in this illustrative example has only one source component comprised of the nodes $1$ and $2$, and node $1$ is a root node for $\lambda=1.5$ (node $1$ can detect $\lambda=1.5$). Thus, the system and graph illustrated in this example satisfy Condition 2, and hence also Condition 1.} 

\section{Estimation Scheme for systems and graphs satisfying Condition $1$}
\label{Estimation}
In this section, we develop a distributed observer for systems and graphs satisfying Condition 1. For presenting the key ideas while reducing notational complexity, we shall make the following assumption.

\begin{assumption}
The graph $\mathcal{G}$ is strongly connected, i.e., there exists a directed path from any node $i$ to any other node $j$, where $i,j \in \mathcal{V}$.
\label{assump:graph}
\end{assumption}

Later, we shall argue that the development can be easily extended to any general directed network. For now, it suffices to say that any directed graph can be decomposed into strong components, some of which are source components (strong components with no incoming edges from the rest of the network); the strategy that we develop here for a strongly connected graph will be employed within each source component. 

\begin{remark}
Since we are focusing on systems and graphs satisfying Condition 1, it follows that under Assumption \ref{assump:graph}, the pair $(\mathbf{A,C})$ is detectable (as a strongly connected graph is one single source component).
\end{remark}

\begin{remark} Note that under Condition $1$ with a strongly connected graph, one might consider the possibility of aggregating all the sensor measurements at a central node and constructing a centralized Luenberger observer, leveraging the fact that $(\mathbf{A,C})$ is detectable. However, for large networks, the routing of measurement information to and from such a central node via multiple hops would induce delays. A distributed approach (such as the one considered in this paper) alleviates such a difficulty. Additionally, as we discuss in Section \ref{sec:timevar}, the general framework developed in this paper allows for extensions to communication losses and sensor failures, which are typical benefits expected of a distributed algorithm.
\end{remark}

We are now in a position to detail the steps to be followed for designing a distributed observer for systems and graphs satisfying Condition 1. We start by providing a generalization of the Kalman observable canonical form to a setting with multiple sensors.
\begin{figure}[t]
\centering
\label{fig:multisensor}
\scalebox{0.7}{
\begin{tikzpicture}
[->,shorten >=1pt,scale=.5, minimum size=10pt, auto=center, node distance=3cm,
  thick]
\tikzstyle{myarrows}=[line width=1mm,draw=black,,postaction={draw, line width=0.5mm, shorten >=1mm, -}]
\node[] at (0,0) {$
\underbrace{\begin{bmatrix} 1&2&-2&-15\\
0&2&4&-16\\
0&0&3&-3\\
0&0&0&0
\end{bmatrix}}_{\mathbf{A}}
$};
\draw [myarrows](3.5,-3)--(4.5,-3);
\node[] at (3.9,-2.2) {$\mathbf{T}_1$};
\node[] at (11.4,0) {$
\underbrace{\left[
\begin{array}{c|ccc} 1&\multicolumn{3}{c}{\mathbf{0}}\\
\hline
-188.18&1.24&6.79&-19.04\\
253.44&0.35&-0.08&6.16\\
76.18&-0.04&-0.78&3.85
\end{array}\right]}_{\bar{\mathbf{A}}_1}
$};
\draw [myarrows](18.5,-3)--(19.5,-3);
\node[] at (18.9,-2.2) {$\mathbf{T}_2$};
\node[] at (26.2,0) {$
\underbrace{\left[
\begin{array}{c|c|cc} 1&\multicolumn{3}{c}{\mathbf{0}}\\
\hline
-81.42&2&\multicolumn{2}{c}{\mathbf{0}}\\
\cline{2-4}
262.1&-14.86&-7.96&12.48\\
84.34&-10.01&-6.99&10.96
\end{array}\right]}_{\bar{\mathbf{A}}_2}
$};
\draw [myarrows](33.2,-3)--(34.2,-3);
\node[] at (33.6,-2.2) {$\mathbf{T}_3$};
\node[] at (40.5,0) {$
\underbrace{\left[
\begin{array}{c|c|c|c} 1&\multicolumn{3}{c}{\mathbf{0}}\\
\hline
-81.42&2&\multicolumn{2}{c}{\mathbf{0}}\\
\cline{2-4}
10.84&0.07&3&\multicolumn{1}{c}{0}\\
\cline{3-4}
266.34&-17.91&-125.33&0
\end{array}\right]}_{\bar{\mathbf{A}}=\bar{\mathbf{A}}_3}
$};
\node[] at (0,-6.7) {$
\underbrace{\begin{bmatrix} 7&-14&35&14\\
0&2&-8&-4\\
0&0&5&-5\\
\end{bmatrix}}_{\mathbf{C}={\begin{bmatrix}\mathbf{C}^T_{1} & \mathbf{C}^T_{2} & \mathbf{C}^T_{3}\end{bmatrix}}^T}
$};
\node[] at (11.4,-6.7) {$
\underbrace{\left[
\begin{array}{c|ccc} 1666&\multicolumn{1}{c}{0}&0&0\\
\end{array}\right]}_{\bar{\mathbf{C}}_1}
$};
\node[] at (26.2,-6.7) {$
\underbrace{\left[
\begin{array}{cc|cc} -364&4.47&\multicolumn{1}{c}{0}&0\\
\end{array}\right]}_{\bar{\mathbf{C}}_2}
$};
\node[] at (40.5,-6.7) {$
\underbrace{\left[
\begin{array}{cccc} 1666&\multicolumn{1}{|c}{0}&0&0\\
\hline
 -364&4.47&\multicolumn{1}{|c}{0}&0\\
 \hline
105&2.94&41.45&\multicolumn{1}{|c}{0}\\
\end{array}\right]}_{\bar{\mathbf{C}}={\begin{bmatrix}\bar{\mathbf{C}}^T_{1} & \bar{\mathbf{C}}^T_{2} & \bar{\mathbf{C}}^T_{3}\end{bmatrix}}^T} 
$};
\node[] at (5,-12.9) {$
\mathbf{T}_1=\begin{bmatrix} 7&0.3430&-0.8575&-0.3430\\
-14&0.8996&0.2511&0.1004\\
35&0.2511&0.3723&-0.2511\\
14&0.1004&-0.2511&0.8996
\end{bmatrix},
$};
\node[] at (20.4,-12.9) {$
\mathbf{T}_2=\begin{bmatrix} 1&0&0&0\\
0&-0.611&-0.6964&-0.6569\\
0&-1.4724&0.6238&-0.3549\\
0&-1.3890&-0.3549&0.6652
\end{bmatrix},
$};
\node[] at (34,-12.9) {$
\mathbf{T}_3=\begin{bmatrix} 1&0&0&0\\
0&1&0&0\\
0&0&3.4616&0.8431\\
0&0&-5.4281&0.5377
\end{bmatrix}
$};
\end{tikzpicture}
}
\centering\mycaption{Illustration of the Multi-Sensor Observable Canonical Decomposition for a detectable pair $\mathbf{(A,C)}$.}
\end{figure}
\subsection{Multi-Sensor Observable Canonical Decomposition}
\label{Transformations}
Given a system matrix $\mathbf{A}$ and a set of $N$ sensors where the $i$-th sensor has an observation matrix given by $\mathbf{C}_i$, we introduce the notion of a \textit{multi-sensor observable canonical decomposition} in this section. The basic philosophy underlying such a decomposition is as follows: given a list of indexed sensors,  perform an observable canonical decomposition with respect to the first sensor. Then, identify the observable portion of the state space with respect to sensor 2 \textit{within the unobservable subspace of sensor 1}, and repeat the process until  the last sensor is reached. Thus, one needs to perform $N$ observable canonical decompositions, one for each sensor, with the last decomposition revealing the portions of the state space that can and cannot be observed using the cumulative measurements of all the sensors. The details of the multi-sensor observable canonical decomposition are captured by the proof of the following result (given in Appendix \ref{sec:proof_multisensor}).
\begin{proposition}
\label{transformations}
Given a system matrix $\mathbf{A}$, and a set of $N$ sensor observation matrices $\mathbf{C}_1, \mathbf{C}_2, \cdots, \mathbf{C}_{N}$, define $\mathbf{C} \triangleq {\begin{bmatrix}\mathbf{C}^T_{1} & \cdots & \mathbf{C}^T_{N}\end{bmatrix}}^T$. Then, there exists a similarity transformation matrix $\mathcal{T}$ which transforms the pair $\mathbf{(A,C)}$ to $(\bar{\mathbf{A}},\bar{\mathbf{C}})$, such that
\begin{equation}
\resizebox{0.6\hsize}{!}{$
\begin{split}
\bar{\mathbf{A}} &= \left[
\begin{array}{c|c|ccc}
\mathbf{A}_{11}  & \multicolumn{4}{c}{\mathbf{0}} \\
\hline
\mathbf{A}_{21}
 & 
\mathbf{A}_{22} & \multicolumn{3}{c}{\mathbf{0}} \\
\cline{2-5}
\vdots & \vdots & \hspace{-5mm} \ddots  & \vdots & \vdots\\
\mathbf{A}_{(N-1)1}& \mathbf{A}_{(N-1)2} \cdots &  \mathbf{A}_{(N-1)(N-1)} & \multicolumn{2}{|c}{\mathbf{0}}\\
\cline{3-5}
\mathbf{A}_{N1}& \mathbf{A}_{N2} \hspace{2mm}\cdots & \mathbf{A}_{N(N-1)} & \multicolumn{1}{|c}{\mathbf{A}_{NN}} & \multicolumn{1}{|c}{\mathbf{0}}\\
\cline{4-5}
\mathbf{A}_1& \hspace{-2mm}\mathbf{A}_2 \hspace{2mm} \cdots & \mathbf{A}_{(N-1)} & \multicolumn{1}{|c}{\mathbf{A}_{N}} & \multicolumn{1}{|c}{\mathbf{A}_{\mathcal{U}}}\\
\end{array}
\right], \\~\\
\bar{\mathbf{C}} &= \begin{bmatrix} \bar{\mathbf{C}}_1 \\ \bar{\mathbf{C}}_2 \vspace{-1mm} \\  \vdots \\ \bar{\mathbf{C}}_N \end{bmatrix} = \left[ \begin{array}{ccccc} \mathbf{C}_{{11}} & \multicolumn{4}{|c}{\mathbf{0}}\\
\hline
\mathbf{C}_{{21}} & \multicolumn{1}{c}{\mathbf{C}_{{22}}} & \multicolumn{3}{|c}{\mathbf{0}}\\
\hline
\vdots&\vdots&\vdots&\vdots&\vdots\\
\mathbf{C}_{{N1}} & \multicolumn{1}{c}{\mathbf{C}_{{N2}}}  & \cdots  \mathbf{C}_{N(N-1)} & \mathbf{C}_{NN}  & \multicolumn{1}{|c}{\mathbf{0}}\\
 \end{array}
 \right].
\end{split}
$}
\label{eqn:gen_form}
\end{equation}
Furthermore, the following properties hold: (i) the pair $(\mathbf{A}_{ii},\mathbf{C}_{ii})$ is observable $\forall i \in \{1,2, \cdots, N\}$; and (ii) the matrix $\mathbf{A}_{\mathcal{U}}$ describes the dynamics of the unobservable subspace of the pair $(\mathbf{A,C})$.
\end{proposition}

Figure 3 illustrates the steps of the multi-sensor observable canonical decomposition for a sensor network with $3$ nodes. The first step involves an observable canonical decomposition of the pair ($\mathbf{A,C}_1$) via the matrix $\mathbf{T}_1$. Next, $\mathbf{T}_2$ reveals the portion of the unobservable subspace of ($\mathbf{A,C}_1$) that can be observed using the observation matrix $\mathbf{C}_2$. Finally, $\mathbf{T}_3$ reveals the portion of the unobservable subspace of ($\mathbf{A},\begin{bmatrix}\mathbf{C}_1\\\mathbf{C}_2\end{bmatrix}$) that can be observed using the observation matrix $\mathbf{C}_3$. For this example, we have $\mathcal{T}=\mathbf{T}_1\mathbf{T}_2\mathbf{T}_3$. In the following section, we discuss how the multi-sensor observable canonical decomposition is applicable to the problem of designing a distributed observer for systems and graphs satisfying Condition 1.

\begin{remark}Note that while describing the multi-sensor observable canonical decomposition, we did not specify any rule for indexing the sensors. This is precisely  because the technique we propose solves Problem 1 regardless of the way the sensors are indexed, as long as the system and graph satisfy Condition 1. However, the question of appropriately ordering the sensors (or including redundancy) will become important when dealing with stochastic systems  or with sensor failures. We discuss these issues later in the paper.
\end{remark}
\subsection{Observer Design}
\label{Observer}
 Using the matrix $\mathcal{T}$ identified in Proposition \ref{transformations}, we perform the coordinate transformation $\mathbf{x}[k]=\mathcal{T}\mathbf{z}[k]$ to obtain
\begin{equation}
\label{eqn:transformed_dyn}
\begin{split}
\mathbf{z}[k+1]&=\bar{\mathbf{A}}\mathbf{z}[k],\\
\mathbf{y}_i[k]&=\bar{\mathbf{C}}_i\mathbf{z}[k], \quad \forall i \in \{1, \cdots, N\},
\end{split}
\end{equation}
where $\bar{\mathbf{A}}={\mathcal{T}}^{-1}\mathbf{A}\mathcal{T}$  and $\bar{\mathbf{C}}_i = \mathbf{C}_i\mathcal{T} = \begin{bmatrix} \mathbf{C}_{i1} & \mathbf{C}_{i2} & \cdots & \mathbf{C}_{i(i-1)} & \mathbf{C}_{ii}  & \multicolumn{1}{|c}{\mathbf{0}}\end{bmatrix}$ are given by (\ref{eqn:gen_form}). The vector $\mathbf{z}[k]$ assumes the following structure (commensurate with the structure of $\bar{\mathbf{A}}$ in (\ref{eqn:gen_form})):
\begin{equation}
\mathbf{z}[k]={\begin{bmatrix}
{\mathbf{z}^{(1)}_{}[k]}^{T}&
\cdots&
{\mathbf{z}^{(N)}_{}[k]}^{T}&
{\mathbf{z}_{\mathcal{U}}[k]}^{T}
\end{bmatrix}}^{T}.
\label{eqn:substates}
\end{equation}
Here, $\mathbf{z}_{\mathcal{U}}[k]$ is precisely the unobservable portion of the state $\mathbf{z}[k]$, with respect to the pair $(\mathbf{A,C})$. We call $\mathbf{z}^{(j)}_{}[k] \in {\mathbb{R}}^{o_j}$ the $j$-th sub-state, and $\mathbf{z}_{\mathcal{U}}[k]$ the unobservable sub-state. Notice that based on the multi-sensor observable canonical decomposition, there is a one-to-one correspondence between a node $j$ and its associated sub-state $\mathbf{z}^{(j)}_{}[k]$. Accordingly, node $j$ is viewed as the source of information of its corresponding sub-state $\mathbf{z}^{(j)}_{}[k]$, and is tasked with the responsibility of estimating this sub-state. For each of the $N$ sub-states, we thus have a unique source of information (based on the initial labeling of the nodes). However, there is no unique source of information for the unobservable sub-state $\mathbf{z}_{\mathcal{U}}[k]$, as this portion of the state does not correspond to the observable subspace of any of the nodes in the network. Each node will thus maintain an estimate of $\mathbf{z}_{\mathcal{U}}[k]$, which it updates as a linear function of its \textit{own} estimates of each of the $N$ sub-states $\mathbf{z}^{(j)}_{}[k], \forall j \in \{1,2, \cdots, N\}$.

\begin{remark} 
It should be noted that a given sub-state $\mathbf{z}^{(j)}_{}[k]$ in equation (\ref{eqn:substates}) might be of zero dimension (i.e., the sub-state can be empty). For instance, this can happen if its corresponding source of information, namely node $j$, has no measurements, i.e., if $\mathbf{C}_j=\mathbf{0}$. 
\end{remark}

First, based on equations (\ref{eqn:gen_form}), (\ref{eqn:transformed_dyn}) and (\ref{eqn:substates}), we observe that the dynamics of the $i$-th sub-state are governed by the equations
\begin{equation}
\begin{split}
\mathbf{z}^{(i)}_{}[k+1]=\mathbf{A}_{ii}\mathbf{z}^{(i)}_{}[k]+\sum_{j=1}^{i-1}\mathbf{A}_{ij}\mathbf{z}^{(j)}_{}[k],\\
\mathbf{y}_i[k]=\mathbf{C}_{ii}\mathbf{z}^{(i)}_{}[k]+\sum_{j=1}^{i-1}\mathbf{C}_{ij}\mathbf{z}^{(j)}_{}[k].
\end{split}
\label{eqn:gendyn}
\end{equation}
The reader is referred to the proof of Proposition \ref{transformations} in Appendix \ref{sec:proof_multisensor} for a mathematical description of the matrices appearing in \eqref{eqn:gendyn}.
 Note that the unobservable sub-state $\mathbf{z}_{\mathcal{U}}[k]$ is governed by the dynamics
\begin{equation}
\mathbf{z}_{\mathcal{U}}[k+1]=\mathbf{A}_{\mathcal{U}}\mathbf{z}_{\mathcal{U}}[k]+\sum_{j=1}^{N}\mathbf{A}_{j}\mathbf{z}^{(j)}_{}[k],
\label{eqn:unobs_state}
\end{equation}
where the matrices $\mathbf{A}_{j}$ describe the coupling that exists between the unobservable sub-state $\mathbf{z}_{\mathcal{U}}[k]$ and each of the $N$ sub-states $\mathbf{z}^{(j)}_{}[k]$. Define $\hat{\mathbf{z}}^{(j)}_{i}[k]$ as the estimate of the $j$-th sub-state maintained by the $i$-th node. The estimation policy adopted by the $i$-th node is as follows - it uses a Luenberger-style update rule for updating its associated sub-state estimate $\hat{\mathbf{z}}^{(i)}_{i}[k]$, and a consensus based scheme for updating its estimates of all other sub-states $\hat{\mathbf{z}}^{(j)}_{i}[k]$, where $ j \in \{1, \cdots ,N\} \setminus\{i\}.$ Based on the dynamics (\ref{eqn:gendyn}), the Luenberger observer at node $i$ is constructed as 
\begin{equation}
\begin{split}
\hat{\mathbf{z}}^{(i)}_{i}[k+1]&=\mathbf{A}_{ii}\hat{\mathbf{z}}^{(i)}_{i}[k]+\sum_{j=1}^{i-1}\mathbf{A}_{ij}\hat{\mathbf{z}}^{(j)}_{i}[k]\\
&\hspace{2mm}+\mathbf{L}_i\left(\mathbf{y}_i[k]-\left(\mathbf{C}_{ii}\hat{\mathbf{z}}^{(i)}_{i}[k]+\sum_{j=1}^{i-1}\mathbf{C}_{ij}\hat{\mathbf{z}}^{(j)}_{i}[k]\right)\right),\\
\end{split}
\label{eqn:luen}
\end{equation}
where $\mathbf{L}_i \in {\mathbb{R}}^{o_i \times r_i}$ is a gain matrix which needs to be designed.
For estimation of the $j$-th sub-state, where $j \in \{1, \cdots ,N\}\setminus\{i\}$, the $i$-th node again mimics the first equation in (\ref{eqn:gendyn}), but this time relies on consensus dynamics of the form
\begin{equation}
\hat{\mathbf{z}}^{(j)}_{i}[k+1]=\underbrace{\mathbf{A}_{jj}\sum_{l\in\mathcal{N}_i}w^{j}_{il}\hat{\mathbf{z}}^{(j)}_{l}[k]}_{\hbox{consensus term}}+\underbrace{\sum_{l=1}^{j-1}\mathbf{A}_{jl}\hat{\mathbf{z}}^{(l)}_{i}[k]}_{\hbox{coupling term}},
\label{eqn:consensus}
\end{equation}
where $w^{j}_{il}$ is the weight the $i$-th node associates with the $l$-th node, for the estimation of the $j$-th sub-state. The weights are non-negative and satisfy
\begin{equation}
\sum_{l\in\mathcal{N}_i}w^{j}_{il}=1, \quad \forall j \in \{1, \cdots ,N\}\setminus\{i\}.
\label{eqn:stochasticity}
\end{equation}
In equation (\ref{eqn:consensus}), the first term is a standard consensus term, while the second term has been introduced specifically to account for the coupling that exists between a given sub-state $j$ and sub-states $1$ to $j-1$ (as given by (\ref{eqn:gendyn})). Let $\hat{\mathbf{z}}_{i\mathcal{U}}[k]$ denote the estimate of the unobservable sub-state $\mathbf{z}_{\mathcal{U}}[k]$ maintained by the $i$-th node. Mimicking equation (\ref{eqn:unobs_state}), each node $i$ uses the following rule to update $\hat{\mathbf{z}}_{i\mathcal{U}}[k]$:
\begin{equation}
\hat{\mathbf{z}}_{i\mathcal{U}}[k+1]=\mathbf{A}_{\mathcal{U}}\hat{\mathbf{z}}_{i\mathcal{U}}[k]+\sum_{j=1}^{N}\mathbf{A}_{j}\hat{\mathbf{z}}^{(j)}_{i}[k].
\label{eqn:unobsestimate}
\end{equation}
In summary, equations (\ref{eqn:luen}), (\ref{eqn:consensus}) and (\ref{eqn:unobsestimate}) together form the observer for the state $\mathbf{z}[k]={\mathcal{T}}^{-1}\mathbf{x}[k]$ maintained by each node $i$.

\subsection{Error Dynamics at the $i$-th Node}
Define $\mathbf{e}^{(j)}_{i}[k] \triangleq \hat{\mathbf{z}}^{(j)}_{i}[k] - \mathbf{z}^{(j)}_{}[k]$ as the error in estimation of the $j$-th sub-state by the $i$-th node. Using equations (\ref{eqn:gendyn}) and (\ref{eqn:luen}), we obtain the error in the Luenberger observer dynamics at the $i$-th node as 
\begin{equation}
\mathbf{e}^{(i)}_{i}[k+1]=\left(\mathbf{A}_{ii}-\mathbf{L}_i\mathbf{C}_{ii}\right)\mathbf{e}^{(i)}_{i}[k]+\sum_{j=1}^{i-1}\left(\mathbf{A}_{ij}-\mathbf{L}_i\mathbf{C}_{ij}\right)\mathbf{e}^{(j)}_{i}[k].
\label{eqn:errluen}
\end{equation}
Similarly, noting that $\mathbf{A}_{jj} = \mathbf{A}_{jj}\sum_{l\in\mathcal{N}_i}w^{j}_{il}$ (based on equation (\ref{eqn:stochasticity})), and using equations (\ref{eqn:gendyn}) and (\ref{eqn:consensus}), we obtain the following consensus error dynamics at node $i$, $\forall j \in \{1, \cdots ,N\}\setminus\{i\}$:
\begin{equation}
\mathbf{e}^{(j)}_{i}[k+1] = \mathbf{A}_{jj}\sum_{l\in \mathcal{N}_i}w^{j}_{il}\mathbf{e}^{(j)}_{l}[k]+\sum_{l=1}^{j-1}\mathbf{A}_{jl}\mathbf{e}^{(l)}_{i}[k].
\label{eqn:errconsensus}
\end{equation}
Define $\mathbf{e}_{i\mathcal{U}}[k] \triangleq \hat{\mathbf{z}}_{i\mathcal{U}}[k]-\mathbf{z}_{\mathcal{U}}[k]$ as the error in estimation of the unobservable sub-state $\mathbf{z}_{\mathcal{U}}[k]$ by the $i$-th node. Using (\ref{eqn:unobs_state}) and (\ref{eqn:unobsestimate}), we obtain the following error dynamics for the unobservable sub-state at node $i$:
\begin{equation}
\mathbf{e}_{i\mathcal{U}}[k+1] = \mathbf{A}_{\mathcal{U}}\mathbf{e}_{i\mathcal{U}}[k]+\sum_{j=1}^{N}\mathbf{A}_{j}\mathbf{e}^{(j)}_{i\mathcal{}}[k].
\label{eqn:errorunobs}
\end{equation}

\subsection{Analysis of the Estimation Scheme for Systems and Graphs Satisfying Condition $1$}
In this section, we present our main result, formally stated as follows.
\begin{theorem}
\label{GEPthm}
Consider a system $(\mathbf{A,C})$ and graph $\mathcal{G}$ satisfying Condition 1. Let Assumption \ref{assump:graph} hold true. Then, for each node $i \in \{1,2, \cdots, N\}$, there exists a choice of observer gain matrix $\mathbf{L}_i$, and consensus weights $w^{j}_{il}$, $j \in \{1,2, \cdots, N\}\setminus\{i\}$, $l \in \mathcal{N}_i$, such that the update rules given by equations (\ref{eqn:luen}), (\ref{eqn:consensus}), and (\ref{eqn:unobsestimate}) form a distributed observer.
\end{theorem}

\begin{proof}
 Consider the composite error in estimation of sub-state $j$ by all of the nodes in $\mathcal{V}$, defined as 
\begin{equation}
\mathbf{E}^{(j)}_{}[k]\triangleq\left[\begin{array}{c}\mathbf{e}^{(j)}_{1}[k]\vspace{2mm}\\
\mathbf{e}^{(j)}_{2}[k]\\
\vdots\\
\mathbf{e}^{(j)}_{N}[k]
\end{array}\right].
\label{eqn:Bigerr}
\end{equation}
We will prove that $\mathbf{E}^{(j)}_{}[k]$ converges to zero asymptotically $\forall j \in \{1, \cdots, N\}$ (recall that there are precisely $N$ nodes in the network, each responsible for estimating a certain sub-state). We prove by induction on $j$. Consider the base case $j=1$, i.e., the estimation of the first sub-state. Let the index set $\{1,k_1, k_2, \cdots, k_{N-1}\}$ represent a topological ordering\footnote{Such an ordering results when a standard Breadth-First Search (BFS) \cite{bondy} algorithm is applied to the graph $\mathcal{G}$, with node $1$ as the root node of the tree. Specifically, the order represents the order in which the nodes are added to the spanning tree when the BFS algorithm is implemented, i.e., node $k_1$ would be added first, followed by node $k_2$ and so on. This ordering naturally leads to a lower triangular adjacency matrix for the constructed spanning tree.} consistent with a spanning tree rooted at node 1 (the source of information for sub-state 1). Note that based on Assumption \ref{assump:graph}, it is always possible to find such a spanning tree. Next, consider the composite error vector
\begin{equation}
\bar{\mathbf{E}}^{(1)}_{}[k]=\left[\begin{array}{c}\mathbf{e}^{(1)}_{1}[k]\vspace{2mm}\\
\mathbf{e}^{(1)}_{k_1}[k]\\
\vdots\\
\mathbf{e}^{(1)}_{k_{N-1}}[k]
\end{array}\right]=\left[\begin{array}{c}\mathbf{e}^{(1)}_{1}[k]\vspace{2mm}\\
{\tilde{\mathbf{E}}}^{(1)}[k]\end{array}\right],
\end{equation}
 where ${\tilde{\mathbf{E}}}^{(1)}[k] \triangleq {\begin{bmatrix} {\mathbf{e}^{(1)}_{k_1}[k]}^{T} \cdots {\mathbf{e}^{(1)}_{k_{N-1}}[k]}^{T} \end{bmatrix}}^{T}.$ Note that $\bar{\mathbf{E}}^{(1)}_{}[k]$ is simply a permutation of the rows of $\mathbf{E}^{(1)}_{}[k]$.
Based on the error dynamics equations given by (\ref{eqn:errluen}) and (\ref{eqn:errconsensus}), we obtain

\[\Scale[0.85]{
\underbrace{\left[\begin{array}{c}\mathbf{e}^{(1)}_{1}[k+1]\vspace{2mm}\\
\tilde{\mathbf{E}}^{(1)}_{}[k+1]\end{array}\right]}_{\bar{\mathbf{E}}^{(1)}_{}[k+1]}=
\underbrace{\left[\begin{array}{cc} (\mathbf{A}_{11}-\mathbf{L}_{1}\mathbf{C}_{11}) & \mathbf{0} \\
\mathbf{W}^{1}_{21} \otimes \mathbf{A}_{11} & \mathbf{W}^{1}_{22} \otimes \mathbf{A}_{11} \end{array}\right]}_{\mathbf{M}_1}\underbrace{\left[\begin{array}{c}\mathbf{e}^{(1)}_{1}[k]\vspace{2mm}\\
\tilde{\mathbf{E}}^{(1)}_{}[k]\end{array}\right]}_{\bar{\mathbf{E}}^{(1)}_{}[k]},
}\]
where the entries of the weight matrix $\mathbf{W}^{1}=\begin{bmatrix}\mathbf{W}^{1}_{21} & \mathbf{W}^{1}_{22}\end{bmatrix}$ are populated by the appropriate weights defined by equation (\ref{eqn:errconsensus}) (note that $\mathbf{W}^{1} \in \mathbb{R}^{(N-1) \times N}$ and $\mathbf{W}^{1}_{21}$ is the first column of $\mathbf{W}^{1}$). Notice that $sp(\mathbf{M}_1)=sp(\mathbf{A}_{11}-\mathbf{L}_{1}\mathbf{C}_{11}) \cup sp(\mathbf{W}^{1}_{22} \otimes \mathbf{A}_{11})$. By construction, the pair $(\mathbf{A}_{11}, \mathbf{C}_{11})$ is observable. Thus, it is always possible to find a gain matrix $\mathbf{L}_1$ such that $(\mathbf{A}_{11}-\mathbf{L}_1\mathbf{C}_{11})$  is Schur stable. Next, we impose the constraint that for the estimation of sub-state 1, non-zero consensus weights are assigned to only the branches of the spanning tree consistent with the ordering $\{1,k_1, k_2, \cdots, k_{N-1}\}$, i.e., a node listens to only its parent in such a tree. In this way, $\mathbf{W}^1_{22}$ becomes lower triangular with eigenvalues equal to zero, without violating the stochasticity condition imposed on $\mathbf{W}^1$ by equation (\ref{eqn:stochasticity}). We conclude that by an appropriate choice of consensus weights, we can achieve $\Lambda_{U}(\mathbf{W}^{1}_{22} \otimes \mathbf{A}_{11}) = \emptyset$ (even if $\Lambda_{U}(\mathbf{A}_{11}) \neq \emptyset)$.\footnote{Here, we use the result that if $\mathbf{A} \in \mathbb{R}^{n \times n}$ and $\mathbf{B} \in \mathbb{R}^{m \times m}$, then the eigenvalues of the Kronecker product $\mathbf{A} \otimes \mathbf{B} \in \mathbb{R}^{mn \times mn}$ are the $mn$ numbers $\lambda_i(\mathbf{A})\lambda_j(\mathbf{B}), (i=1, \cdots, n; j=1, \cdots, m)$ \cite{matrixsurvey}.}

Thus, $\mathbf{M}_1$ can be made Schur stable and hence $\lim_{k\to\infty} \bar{\mathbf{E}}^{(1)}_{}[k]=\mathbf{0}$, implying $\lim_{k\to\infty} \mathbf{E}^{(1)}_{}[k] = \mathbf{0} $ (one is just a permutation of the other). Thus, the base case is proven. Next, suppose that $\mathbf{E}^{(j)}_{}[k] $ converges to zero asymptotically $\forall j \in \{1, \cdots ,p-1\}$, where $1 \leq p-1 \leq N-1$. Consider the following composite error vector for the $p$-th sub-state:
\begin{equation}
\bar{\mathbf{E}}^{(p)}_{}[k]=\left[\begin{array}{c}\mathbf{e}^{(p)}_{p}[k]\vspace{2mm}\\
\mathbf{e}^{(p)}_{m_1}[k]\\
\vdots\\
\mathbf{e}^{(p)}_{m_{N-1}}[k]
\end{array}\right]=\left[\begin{array}{c}\mathbf{e}^{(p)}_{p}[k]\vspace{2mm}\\
\tilde{\mathbf{E}}^{(p)}_{}[k]\end{array}\right],
\end{equation}
where the index set $\{p, m_1, m_2, \cdots, m_{N-1}\}$ represents a topological ordering of the nodes of $\mathcal{V}$ to obtain a spanning tree rooted at node $p$ (the source of information for sub-state $p$), and $\tilde{\mathbf{E}}^{(p)}_{}[k]  \triangleq {\begin{bmatrix} {\mathbf{e}^{(p)}_{m_1}[k]}^{T} \cdots {\mathbf{e}^{(p)}_{m_{N-1}}[k]}^{T} \end{bmatrix}}^{T}.$
From the error dynamics equations given by (\ref{eqn:errluen}) and (\ref{eqn:errconsensus}), we obtain
\begin{equation}
\bar{\mathbf{E}}^{(p)}_{}[k+1]=\mathbf{M}_{p}\bar{\mathbf{E}}^{(p)}_{}[k]+\sum_{l=1}^{p-1} \mathbf{H}_{pl}\bar{\mathbf{E}}^{(pl)}_{}[k],
\end{equation}
where 
\begin{equation}
\mathbf{M}_{p}=\left[\begin{array}{cc} (\mathbf{A}_{pp}-\mathbf{L}_{p}\mathbf{C}_{pp}) & \mathbf{0} \\
\mathbf{W}^{p}_{21} \otimes \mathbf{A}_{pp} & \mathbf{W}^{p}_{22} \otimes \mathbf{A}_{pp} \end{array}\right],
\label{eqn:bigerrmatrix}
\end{equation}
\begin{equation}
\mathbf{H}_{pl}=diag\left(\mathbf{A}_{pl}-\mathbf{L}_{p}\mathbf{C}_{pl},\mathbf{I}_{N-1}\otimes\mathbf{A}_{pl}\right),
\end{equation}
\begin{equation}
\bar{\mathbf{E}}^{(pl)}_{}[k]=\left[\begin{array}{c}\mathbf{e}^{(l)}_{p}[k]\vspace{2mm}\\
\mathbf{e}^{(l)}_{m_1}[k]\\
\vdots\\
\mathbf{e}^{(l)}_{m_{N-1}}[k]
\end{array}\right].
\vspace{4mm}
\end{equation}
By following the same train of logic as the base case, one concludes that $\mathbf{M}_{p}$ can be made Schur stable by appropriate choices of the observer gain matrix $\mathbf{L}_{p}$, and consensus weight matrix ${\mathbf{W}}^{p} = \begin{bmatrix} {\mathbf{W}}^{p}_{21} & {\mathbf{W}}^{p}_{22} \end{bmatrix}$ (note that $\mathbf{W}^{p} \in \mathbb{R}^{(N-1) \times N}$ and $\mathbf{W}^{p}_{21}$ is the first column of $\mathbf{W}^{p}$). Specifically, non-zero weights are assigned in $\mathbf{W}^{p}$ only on the branches of the tree rooted at node $p$, consistent with the topological ordering. Notice that $\bar{\mathbf{E}}^{(pl)}_{}[k]$ is simply a permutation of the rows of $\mathbf{E}^{(l)}_{}[k]$ (permuted to match the order of indices in $\bar{\mathbf{E}}^{(p)}_{}[k]$). Further, based on our induction hypothesis, $\mathbf{E}^{(l)}_{}[k]$ converges to zero asymptotically (since $1 \leq l \leq p-1$). Thus, by Input to State Stability (ISS), we conclude that $\bar{\mathbf{E}}^{(p)}_{}[k],$ and hence $\mathbf{E}^{(p)}_{}[k]$, converges to zero asymptotically. We have thus proven that the composite estimation error for every sub-state asymptotically approaches zero, i.e., $\lim_{k\to\infty}  \mathbf{e}^{(j)}_{i}[k] = \mathbf{0}, \forall i,j \in \{1, \cdots N \}$. 

Finally, consider the error in estimation of the unobservable sub-state $\mathbf{z}_{\mathcal{U}}[k]$ (given by equation (\ref{eqn:errorunobs})). As the system and graph under consideration satisfy Condition 1 and Assumption \ref{assump:graph}, it must be that the pair $(\mathbf{A,C})$ is detectable. Thus, based on Proposition \ref{transformations},  the matrix $\mathbf{A}_{\mathcal{U}}$ in (\ref{eqn:errorunobs}) must be stable. Invoking ISS, we have that $\lim_{k\to\infty} \mathbf{e}_{i\mathcal{U}}[k] = \mathbf{0}, \forall i \in \{1, \cdots, N\}$. Thus, every node in the network can asymptotically estimate $\mathbf{z}[k]$, and hence $\mathbf{x}[k]$, as $\mathbf{x}[k]=\mathcal{T}\mathbf{z}[k]$.
\end{proof}

\subsection{A Compact Representation of the Proposed Observer}
\label{sec:compact}
 In this section, we combine the update equations (\ref{eqn:luen}), (\ref{eqn:consensus}) and (\ref{eqn:unobsestimate}) to obtain a compact representation of our distributed observer. To do so, we need to first introduce some notation. Accordingly, let $\mathbf{B}_j =\begin{bmatrix} \mathbf{0} \cdots \mathbf{I}_{o_j} \cdots \mathbf{0} \end{bmatrix}$ be the matrix that extracts the $j$-th sub-state from the transformed state vector $\mathbf{z}[k]$, i.e., $\mathbf{z}^{(j)}_{}[k]=\mathbf{B}_j\mathbf{z}[k]$. Similarly, let $\mathbf{B}_{\mathcal{U}}$ be such that $\mathbf{z}_{\mathcal{U}}[k]=\mathbf{B}_{\mathcal{U}}\mathbf{z}[k]$. Define $\mathbb{B} \triangleq diag(\mathbf{B}_1, \cdots, \mathbf{B}_N, \mathbf{B}_{\mathcal{U}})$. Next, notice that the transformed system matrix $\bar{\mathbf{A}}$ in equation (\ref{eqn:gen_form}) can be written as $\bar{\mathbf{A}} = \bar{\mathcal{A}}_1+\bar{\mathcal{A}}_2$, where $\bar{\mathcal{A}}_2=diag(\mathbf{A}_{11}, \cdots, \mathbf{A}_{NN}, \mathbf{A}_{\mathcal{U}})$, and $\bar{\mathcal{A}}_1$ is a block lower-triangular matrix given by $\bar{\mathbf{A}} - \bar{\mathcal{A}}_2$. Let $\mathbf{w}_{il}$ (where $l\in\mathcal{N}_i\setminus\{i\}$) be the vector of weights node $i$ associates with a neighbor $l$ for the estimation of the transformed state $\mathbf{z}[k]$. Based on our estimation scheme, note that at any given time-step $k$, node $i$ does not use the estimates received from its neighbors at time-step $k$ for estimating $\mathbf{z}^{(i)}_{}[k]$ and $\mathbf{z}_{\mathcal{U}}[k]$, and hence these weight vectors assume the following form: $\mathbf{w}_{il}={\begin{bmatrix} w^{1}_{il}, \cdots, w^{i-1}_{il}, 0, w^{i+1}_{il}, \cdots, w^{N}_{il}, 0 \end{bmatrix}}^{T}, \forall l\in\mathcal{N}_i\setminus\{i\}$. Also, notice that the element $w^{j}_{il}$ is not present in the vector if the $j$-th sub-state is empty (i.e., of dimension $0$). Similarly, let $\mathbf{w}_{ii}$ be a vector with a `$1$' in the elements corresponding to the $i$-th sub-state and the unobservable sub-state $\mathbf{z}_{\mathcal{U}}[k]$, and zeroes at all other positions. Finally, defining $\mathbb{H}_i \triangleq {\begin{bmatrix} {\mathbf{0}}^{T}  \cdots  {\mathbf{L}_i}^{T}  \cdots  {\mathbf{0}}^{T} \end{bmatrix}}^{T}$, using equations (\ref{eqn:luen}), (\ref{eqn:consensus}) and (\ref{eqn:unobsestimate}),  and noting that $\mathbf{z}[k]={\mathcal{T}}^{-1}\mathbf{x}[k]$, we obtain the following overall state estimate update rule at node $i$:
\begin{equation}
\boxed{
\hat{\mathbf{x}}_i[k+1]=\mathcal{T}\bar{\mathcal{A}}_1 
{\mathcal{T}}^{-1}\hat{\mathbf{x}}_i[k]+\underbrace{{\mathcal{T}}\mathbb{H}_i(\mathbf{y}_i[k]-{\mathbf{C}}_i\hat{\mathbf{x}}_{i}[k])}_{\hbox{innovation term}}
+\underbrace{\sum_{l\in\mathcal{N}_i}\mathbb{G}_{il}\hat{\mathbf{x}}_{l}[k],}_{\hbox{``consensus term''}}}
\label{eqn:Overallest}
\end{equation}
where $\hat{\mathbf{x}}_{i}[k]$ denotes the estimate of the state $\mathbf{x}[k]$ maintained by node $i$, and $\mathbb{G}_{il}={\mathcal{T}}\bar{\mathcal{A}}_2\mathbb{B}\left(\mathbf{w}_{il}\otimes{\mathcal{T}}^{-1}\right).$
\begin{remark}
From the structure of our overall estimator at node $i$, as represented by equation (\ref{eqn:Overallest}), it is easy to see that the estimator maintained at each node has dimension equal to $n$ (i.e.,  equal to that of the state). Thus, our approach alleviates the need to construct augmented observers such as those considered in \cite{martins3,wang}.
\end{remark}

\begin{remark} Note that all the transformation and gain matrices appearing in   \eqref{eqn:Overallest} can be computed offline during a centralized design phase. Thus, although the observer design and the subsequent analysis were done in the $\mathbf{z}[k]$ coordinate system, no inversion from $\mathbf{z}[k]$ to $\mathbf{x}[k]$ is necessary while implementing \eqref{eqn:Overallest} during run-time, i.e., the nodes directly exchange their estimates of the actual state $\mathbf{x}[k]$, and not $\mathbf{z}[k]$.
\end{remark}

\subsection{Summary of the Estimation Scheme for Systems and Graphs Satisfying Condition $1$}
\label{sec:summary1}
The proposed distributed observer scheme for systems and graphs satisfying Condition $1$ (under the assumption that the graph $\mathcal{G}$ is strongly connected) can be broadly decomposed into two main phases, namely the design phase and the distributed estimation phase. For clarity, we briefly enumerate the steps associated with each of these phases.
 
\textbf{Design Phase:} 
 
 \begin{itemize}
 \item Each node of the graph is assigned a unique integer between $1$ to $N$. Based on this numbering, the multi-sensor observable canonical decomposition (as outlined in the proof of Proposition \ref{transformations}) is performed, yielding the state $\mathbf{z}[k]={\mathcal{T}}^{-1}\mathbf{x}[k]$.
 \item Based on this transformation, each node is associated with a sub-state of $\mathbf{z}[k]$ that it is responsible for estimating. Recall that there are precisely $N$ sub-states, one corresponding to each node in the network; some of these sub-states might be empty. 
 \item For the estimation of a given sub-state, we construct a spanning tree rooted at the specific node which acts as the source of information for that sub-state. The resulting spanning tree guides the construction of the consensus weight matrix to be used for the estimation of that particular sub-state. We construct one spanning tree for the estimation of each non-empty sub-state.
\item Based on the constructed consensus weight matrices, and the Luenberger observer gains $\mathbf{L}_i$, the matrices $\mathcal{T}\bar{\mathcal{A}}_1\mathcal{T}^{-1}$, $\mathcal{T}\mathbb{H}_i$ and $\mathbb{G}_{il}$ in \eqref{eqn:Overallest} are computed for each node $i\in\mathcal{V}$.
\end{itemize}
\textbf{Estimation Phase (Run-time):}
\begin{itemize} 
 \item Each node employs a Luenberger observer for constructing an estimate of its corresponding sub-state, and runs consensus dynamics for estimating the sub-states corresponding to the remaining nodes in the network. Summarily, a node implements \eqref{eqn:Overallest} for estimating $\mathbf{x}[k]$.
\end{itemize}

\begin{remark} 
While the observer design procedure we have outlined (involving the multi-sensor decomposition, design of local observer gains, construction of spanning trees and selection of consensus weights) can be readily implemented in a centralized manner, it may also be possible to perform these steps in a distributed fashion.  This would require the nodes to assign themselves unique identifiers (or labels) and execute the multi-sensor decomposition in a round-robin fashion, followed by a distributed construction of spanning trees.  However, at present, the multi-sensor decomposition appears to be the most expensive portion of such an implementation (in terms of coordination and communication).  In Section \ref{sec:Case1}, we show that for systems and graphs that possess the additional structure described by Condition 2, we can avoid such a decomposition and obtain a scheme that permits an efficient distributed implementation (in both the design and run-time phases) at the potential cost of increasing the dimension of the observer.
\label{rem:centralized_design}
\end{remark}

Having established our approach for all systems and strongly connected graphs satisfying Condition 1, we now briefly describe the extension of our strategy to arbitrary directed networks.
 
 \subsection{Extension to General Directed Networks}
Our distributed observer design can be extended to general networks (satisfying Condition 1 but not necessarily Assumption \ref{assump:graph}) by first decomposing $\mathcal{G}$ into its strong components, and identifying each of the source components. Next, within a given source component, one simply follows the observer design procedure outlined in Section \ref{Observer} for a strongly connected graph, to obtain an estimator of the form  (\ref{eqn:Overallest}) for each node within the source component. Define $\mathcal{S} \triangleq \bigcup_{i=1}^{p} \mathcal{V}_i$ to be the set of all nodes that belong to the source components of $\mathcal{G}$. Let each node in $i\in\mathcal{V}\setminus\mathcal{S}$ employ a pure consensus strategy of the form
\begin{equation}
\hat{\mathbf{x}}_{i}[k+1]=\mathbf{A}\sum_{j\in\mathcal{N}_i}w_{ij}\hat{\mathbf{x}}_{j}[k],
\label{eqn:consnonsource}
\end{equation}
where $\hat{\mathbf{x}}_{i}[k]$ represents an estimate of the state maintained by the $i$-th node. The weights $w_{ij}$ are non-negative and satisfy
\begin{equation}
\sum_{j\in\mathcal{N}_i}w_{ij}=1, \quad \forall i \in  \mathcal{V}\setminus\mathcal{S}.
\label{eqn:stochasticnonsource}
\end{equation} 
 The design of consensus weights for the nodes in $\mathcal{V}\setminus\mathcal{S}$ is based on the observation that the set $\mathcal{V}\setminus\mathcal{S}$ can be spanned by a disjoint union of trees rooted in $\mathcal{S}$. By assigning consensus weights to only the branches of these trees (without violating the stochasticity condition imposed by equation (\ref{eqn:stochasticnonsource})), one obtains stable estimation error dynamics for each of the nodes in $\mathcal{V}\setminus\mathcal{S}$ (the details are similar to the proof of Theorem \ref{GEPthm}).
The above strategy readily leads to the following result.

\begin{theorem}
\label{Thmgeneral}
Consider a system $(\mathbf{A,C})$ and graph $\mathcal{G}$ satisfying Condition 1. Let each node in $\mathcal{S}$ run an observer of the form (\ref{eqn:Overallest}), and each node in  $\mathcal{V}\setminus\mathcal{S}$ run the consensus dynamics given by (\ref{eqn:consnonsource}). Then, there exists a choice of consensus weights and observer gain matrices that result in a distributed observer.
\end{theorem}

As discussed in Remark~\ref{rem:centralized_design}, our distributed observer design starts with the multi-sensor observable decomposition described in Proposition~\ref{transformations}, which transforms the system into a form that identifies the sub-states that each node is responsible for estimating.  This decomposition requires knowledge of the measurement matrices of each node, and thus is most amenable to a centralized implementation.\footnote{A centralized design phase is also commonly assumed in the existing literature on distributed observers, e.g., \cite{Khanobs2,ugrinov,wang,martins1,martins2,martins3}.}  In the next section, we show that for systems and networks satisfying Condition 2, the design of the observer itself can be readily done in a distributed manner.  

\section{Estimation Scheme for systems and graphs satisfying Condition $2$}
\label{sec:Case1}
Recall that for systems and graphs satisfying Condition 2, for each eigenvalue of the plant, there is at least one node in each source component that can detect that eigenvalue. As we will show, this fact allows each node in the network to identify the sub-states it is responsible for estimating, without having to exchange any information with neighbors.  

To this end, let $\mathcal{T}$ be a non-singular transformation matrix which transforms $\mathbf{A}$ into its Jordan canonical form $\mathbf{J}$, i.e., $\mathbf{A}=\mathcal{T}\mathbf{J}\mathcal{T}^{-1}$. With $\mathbf{z}[k]=\mathcal{T}^{-1}\mathbf{x}[k]$, the dynamics (\ref{eqn:plant}) are transformed into the form
\begin{equation}
\begin{split}
\mathbf{z}[k+1]&=\mathbf{Jz}[k],\\
\mathbf{y}_i[k]&=\bar{\mathbf{C}}_i\mathbf{z}[k], \quad \forall {i} \in \{1, \cdots, N\}
\end{split}
\label{eqn:REPsystem}
\end{equation}
where $\mathbf{J}=\mathcal{T}^{-1}\mathbf{A}\mathcal{T}$ and $\bar{\mathbf{C}}_i=\mathbf{C}_i\mathcal{T}$.\footnote{Note that the matrices $\mathcal{T}$ and $\bar{\mathbf{C}}_i$ in (\ref{eqn:REPsystem}) are in general different from those in (\ref{eqn:gen_form}); we adopt this abuse of notation to avoid cluttering the exposition with additional symbols.} Notice that this transformation relies only on the knowledge of the system matrix $\mathbf{A}$ (which is assumed to be known by all of the nodes). Hence, all nodes can perform this transformation {\it in parallel} (e.g., by using an agreed-upon convention for ordering the eigenvalues and corresponding eigenvectors). We denote the eigenvalues of $\mathbf{J}$ (which are the same as those of $\mathbf{A}$) by $\lambda_1, \cdots, \lambda_{\gamma}$, where $\gamma$ represents the number of distinct eigenvalues of $\mathbf{A}$. Let $\mathbf{J}=diag(\mathbf{J}_1, \cdots, \mathbf{J}_{\gamma})$, where we group all of the Jordan blocks associated with  $\lambda_j \in  sp(\mathbf{J}) $ into the block diagonal matrix $\mathbf{J}_j \in \mathbb{R}^{a_{\mathbf{J}}(\lambda_j) \times a_{\mathbf{J}}(\lambda_j)}$. The portion of the state $\mathbf{z}[k]$ associated with the eigenvalue $\lambda_j$ is termed as the sub-state $\mathbf{z}^{(j)}[k] \in {\mathbb{R}}^{a_\mathbf{J}(\lambda_j)}$. Let $\hat{\mathbf{z}}^{(j)}_{i}[k]$ represent the estimate of $\mathbf{z}^{(j)}[k]$ maintained by node $i$. Note that if each node in the network can accurately estimate $\mathbf{z}[k]$, then they can also estimate $\mathbf{x}[k]$ using the relation $\mathbf{x}[k]=\mathcal{T}\mathbf{z}[k]$. In view of this, we now develop a scheme for estimating $\mathbf{z}[k]$. 

\subsection{Distributed Observer Design}
\subsubsection{Design of Local Luenberger Observers} Let $\mathcal{O}_i$  represent the set of detectable eigenvalues of node $i$.  For estimating the sub-states corresponding to the eigenvalues in $\mathcal{O}_i$, node $i$ constructs a simple Luenberger observer using its own measurements. To this end, permute the states $\mathbf{z}[k]$ in \eqref{eqn:REPsystem} to obtain
 \begin{equation}
\begin{split}
 \underbrace{\begin{bmatrix} \mathbf{z}_{\mathcal{O}_i}[k+1] \\ \mathbf{z}_{\mathcal{UO}_i}[k+1] \end{bmatrix}}_{\bar{\mathbf{z}}_i[k+1]} &= \underbrace{\begin{bmatrix} 
\bar{\mathbf{J}}_{\mathcal{O}_i} & \mathbf{0} \\
\mathbf{0}  & \bar{\mathbf{J}}_{\mathcal{UO}_i} 
\end{bmatrix}}_{\bar{\mathcal{J}}_i} \underbrace{\begin{bmatrix} \mathbf{z}_{\mathcal{O}_i}[k] \\ \mathbf{z}_{\mathcal{UO}_i}[k] \end{bmatrix}}_{\bar{\mathbf{z}}_i[k]}, \\
\mathbf{y}_i[k] &= \underbrace{\begin{bmatrix} \bar{\mathbf{C}}_{\mathcal{O}_i} & \bar{\mathbf{C}}_{\mathcal{UO}_i} \end{bmatrix}}_{\bar{\mathcal{C}}_i}\bar{\mathbf{z}}_i[k]. \\
\end{split}
\label{eqn:Case1form}
\end{equation}
The permuted state $\bar{\mathbf{z}}_i[k]$ will be represented by   $\mathbf{z}[k]=\mathbb{P}_i\bar{\mathbf{z}}_i[k]$, where $\mathbb{P}_i$ is an appropriate permutation matrix. In the above equations,  $\bar{\mathbf{J}}_{\mathcal{O}_i}$ consists of all Jordan blocks corresponding to the detectable eigenvalues of node $i$, and $\bar{\mathbf{J}}_{\mathcal{UO}_i}$ denotes the collection of Jordan blocks corresponding to the undetectable eigenvalues of node $i$.
Similarly, $\bar{\mathbf{C}}_{\mathcal{O}_i}$ contains the columns of $\bar{\mathbf{C}}_i$ corresponding to the matrix $\bar{\mathbf{J}}_{\mathcal{O}_i}$, with an analogous definition for $\bar{\mathbf{C}}_{\mathcal{UO}_i}$.
The sub-states corresponding to the detectable and undetectable eigenvalues of node $i$ are grouped into the composite vectors $\mathbf{z}_{\mathcal{O}_i}[k] \in {\mathbb{R}}^{o_i}$ and  $\mathbf{z}_{\mathcal{UO}_i}[k]$ respectively.

Based on \eqref{eqn:Case1form}, notice that the output $\mathbf{y}_i[k]$ is affected by elements of $\mathbf{z}_{\mathcal{UO}_i}[k]$ (through $\bar{\mathbf{C}}_{\mathcal{UO}_i}$) and thus we will estimate those elements as well in order to recover $\mathbf{z}_{\mathcal{O}_i}[k]$. To this end, let $\bar{\mathbf{T}}_i$ be a non-singular matrix which performs an observable canonical decomposition of the pair $(\bar{\mathbf{J}}_{\mathcal{UO}_i},\bar{\mathbf{C}}_{\mathcal{UO}_i})$ in \eqref{eqn:Case1form}. Consider the following transformation matrix:
\begin{equation}
\mathbf{T}_i=\begin{bmatrix} \mathbf{I}_{o_i} & \mathbf{0} \\ \mathbf{0} & \bar{\mathbf{T}}_i \end{bmatrix}.
\label{eqn:Tbari}
\end{equation}
Define the coordinate transformation $\bar{\mathbf{z}}_i[k]=\mathbf{T}_i\mathbf{v}_i[k]$ (this transformation is specific to node $i$). Based on this transformation, and equations (\ref{eqn:Case1form}) and (\ref{eqn:Tbari}), the dynamics at node $i$ can be reformulated as
\begin{equation} 
 \begin{split}
 \underbrace{\left[\begin{array}{c}
\mathbf{z}_{\mathcal{O}_i}[k+1]\\
\mathbf{w}_{\mathcal{O}_i}[k+1]\\
\mathbf{w}_{\mathcal{UO}_i}[k+1]\\
\end{array}\right]}_{\mathbf{v}_i[k+1]} &= 
\underbrace{\left[
\begin{array}{c|cc}
\bar{\mathbf{J}}_{\mathcal{O}_i}  & \multicolumn{2}{c}{\mathbf{0}} \\
\hline
\multirow{2}{*}{$\mathbf{0}$}
 & 
\mathbf{G}_{\mathcal{O}_i} & \mathbf{0} \\
&
\star & \mathbf{G}_{\mathcal{UO}_i}\\
\end{array}
\right]}_{\mathbf{T}^{-1}_i\bar{\mathcal{J}}_i\mathbf{T}_i}\underbrace{\left[\begin{array}{c}
\mathbf{z}_{\mathcal{O}_i}[k]\\
\mathbf{w}_{\mathcal{O}_i}[k]\\
\mathbf{w}_{\mathcal{UO}_i}[k]\\
\end{array}\right]}_{\mathbf{v}_i[k]}, \\
\mathbf{y}_i[k] &= \hspace{6.5mm} \underbrace{\left[
\begin{array}{c|cc}
\bar{\mathbf{C}}_{\mathcal{O}_i} & \mathbf{H}_{\mathcal{O}_i} & \mathbf{0} \end{array}\right]}_{\bar{\mathcal{C}}_i\mathbf{T}_i}\mathbf{v}_i[k],
\end{split}
\label{eqn:Case1obs1}
\end{equation}
where
\begin{equation}
\begin{split}
\bar{\mathbf{T}}^{-1}_i\bar{\mathbf{J}}_{\mathcal{UO}_i}\bar{\mathbf{T}}_i &= \begin{bmatrix} \mathbf{G}_{\mathcal{O}_i} & \mathbf{0} \\
\star & \mathbf{G}_{\mathcal{UO}_i}\end{bmatrix}, \\
\bar{\mathbf{C}}_{\mathcal{UO}_i}\bar{\mathbf{T}}_i &=\begin{bmatrix} \mathbf{H}_{\mathcal{O}_i} & \mathbf{0}\end{bmatrix}.
\end{split}
\label{eqn:transCond2}
\end{equation}
Define 
\begin{equation}
\mathbb{J}_i \triangleq diag(\bar{\mathbf{J}}_{\mathcal{O}_i},  \mathbf{G}_{\mathcal{O}_i}), \enspace \mathbb{F}_i \triangleq \begin{bmatrix} \bar{\mathbf{C}}_{\mathcal{O}_i} & \mathbf{H}_{\mathcal{O}_i} \end{bmatrix},
\label{eqn:JiCi}
\end{equation}
 and $\mathbf{s}_i[k] \triangleq {\begin{bmatrix} {\mathbf{z}_{\mathcal{O}_i}}^T[k] & {\mathbf{w}_{\mathcal{O}_i}}^T[k]\end{bmatrix}}^T$. Based on the dynamics (\ref{eqn:Case1obs1}), the  local Luenberger observer maintained by node $i$ for estimating $\mathbf{z}_{\mathcal{O}_i}[k]$ has the form
\begin{equation}
\hat{\mathbf{s}}_i[k+1]=\mathbb{J}_i\hat{\mathbf{s}}_i[k]+\mathbb{L}_i(\mathbf{y}_i[k]-\mathbb{F}_i\hat{\mathbf{s}}_i[k]),
\label{eqn:luenCase1}
\end{equation}
where $\mathbb{L}_i$ is a gain matrix which needs to be designed for node $i$ and $\hat{\mathbf{s}}_i[k]$ is the estimate of $\mathbf{s}_i[k]$ maintained by node $i$. Using (\ref{eqn:luenCase1}), $\hat{\mathbf{z}}_{\mathcal{O}_i}[k]$ can then be updated as $\hat{\mathbf{z}}_{\mathcal{O}_i}[k+1] = \begin{bmatrix}\mathbf{I}_{o_i} & \mathbf{0} \end{bmatrix} \hat{\mathbf{s}}_i[k+1]$.

Based on the (local) transformation \eqref{eqn:Case1obs1} and the (local) observer \eqref{eqn:luenCase1}, we obtain the following result.

\begin{lemma}
\label{lemmaSEP}
For a system $(\mathbf{A,C})$ and graph $\mathcal{G}$ satisfying Condition $2$, let every node $i \in \mathcal{V}$ run a Luenberger observer of the form (\ref{eqn:luenCase1}). Then, there exists a choice of observer gain $\mathbb{L}_i$, which can be designed locally, such that for each $\lambda_j \in \mathcal{O}_{i}$, $\lim_{k\to\infty}||\hat{\mathbf{z}}^{(j)}_{i}[k]-\mathbf{z}^{(j)}[k]||=0$.
\end{lemma}

\begin{proof}
The proof follows straightforwardly  by noting that $(\mathbb{J}_i,\mathbb{F}_i)$ defined in \eqref{eqn:JiCi} is detectable (since $\bar{\mathbf{J}}_{\mathcal{O}_i}$ and $\mathbf{G}_{\mathcal{O}_i}$ do not share any eigenvalues, and each of the pairs $(\bar{\mathbf{J}}_{\mathcal{O}_i}, \bar{\mathbf{C}}_{\mathcal{O}_i})$ and $(\mathbf{G}_{\mathcal{O}_i}, \mathbf{H}_{\mathcal{O}_i})$ are detectable, by construction).  
\end{proof}

Having established that each node $i \in \mathcal{V}$ can asymptotically recover $\mathbf{z}_{\mathcal{O}_i}[k]$ in \eqref{eqn:Case1form} purely locally, we now devise a method that allows each node to estimate the sub-states corresponding to the locally undetectable eigenvalues. 

\subsubsection{Consensus dynamics} 
Consider an eigenvalue $\lambda_j \in \mathcal{UO}_i$ (recall $\mathcal{UO}_i$ represents the set of undetectable eigenvalues of node $i$). For such an eigenvalue, node $i$ has to rely on the information received from its neighbors in order to estimate $\mathbf{z}^{(j)}[k]$. To this end, we propose the following consensus strategy to be followed by every node $i \in \mathcal{V}\setminus\mathcal{S}_j$ for updating their respective estimates of $\mathbf{z}^{(j)}[k]$:\footnote{Recall that $\mathcal{S}_j$ denotes the set of root nodes that can detect $\lambda_j$.}
\begin{equation}
\hat{\mathbf{z}}^{(j)}_{i}[k+1]=\mathbf{J}_j\sum_{l\in\mathcal{N}_i}w^{j}_{il}\hat{\mathbf{z}}^{(j)}_{l}[k],
\label{eqn:consensusREP}
\end{equation}
where $w^{j}_{il}$ is the weight the $i$-th node associates with the $l$-th node for the estimation of the $j$-th sub-state. Each weight is non-negative and satisfies
\begin{equation}
\sum_{l\in\mathcal{N}_i}w^{j}_{il}=1, \quad \forall \lambda_j \in \mathcal{UO}_{i}.
\label{eqn:REPstochasticity}
\end{equation}
Let $\mathcal{UO}_i = \{\lambda_{n_1}, \cdots, \lambda_{n_{\gamma_i}}\}$, where $\gamma_i = |\mathcal{UO}_i| = \gamma - \sigma_i$ (recall $\sigma_i=|\mathcal{O}_i|$, and $\gamma$ is the number of distinct eigenvalues of $\mathbf{A}$). Define $\mathbf{B}_j =\begin{bmatrix} \mathbf{0} \cdots \mathbf{I}_{o_j} \cdots \mathbf{0} \end{bmatrix}$ as the matrix which extracts the $j$-th sub-state from the state vector $\mathbf{z}[k]$, i.e., we have $\mathbf{z}^{(j)}[k]=\mathbf{B}_j\mathbf{z}[k]$. Also, let $\mathbf{w}_{il} = {\begin{bmatrix} w^{n_1}_{il} \cdots w^{n_{\gamma_i}}_{il} \end{bmatrix}}^{T}$ denote the vector of consensus weights the $i$-th node assigns to the $l$-th node ($l \in \mathcal{N}_i$) for the estimation of the sub-states corresponding to each of its undetectable eigenvalues. Then, noting the definition of $\bar{\mathbf{J}}_{\mathcal{UO}_i}$ and using the consensus equation given by (\ref{eqn:consensusREP}), we obtain
\begin{equation}
\hat{\mathbf{z}}_{\mathcal{UO}_i}[k+1]=\bar{\mathbf{J}}_{\mathcal{UO}_i}\mathbb{B}_i \sum_{l\in\mathcal{N}_i}\mathbf{w}_{il} \otimes \hat{\mathbf{z}}_{l}[k],
\label{eqn:compeq3}
\end{equation}
where $\mathbb{B}_i=diag(\mathbf{B}_{n_1}, \cdots, \mathbf{B}_{\gamma_i})$. Noting that $\mathbf{x}[k]=\mathcal{T}\mathbf{z}[k]$, and $\mathbf{z}[k]=\mathbb{P}_i\bar{\mathbf{z}}_i[k]$ (recall $\bar{\mathbf{z}}_i[k]={\begin{bmatrix} {\mathbf{z}_{\mathcal{O}_i}[k]}^{T} & {\mathbf{z}_{\mathcal{UO}_i}[k]}^{T} \end{bmatrix}}^{T}$), and using equations (\ref{eqn:luenCase1}) and (\ref{eqn:compeq3}), we obtain the governing equations of the distributed observer maintained at node $i$ as
\begin{subequations}
\begin{empheq}[box=\widefbox]{align}
\hat{\mathbf{s}}_i[k+1]&=\mathbb{J}_i\hat{\mathbf{s}}_i[k]+\mathbb{L}_i(\mathbf{y}_i[k]-\mathbb{F}_i\hat{\mathbf{s}}_i[k]),\\
\hat{\mathbf{z}}_{\mathcal{O}_i}[k+1]&= \begin{bmatrix}\mathbf{I}_{o_i} & \mathbf{0} \end{bmatrix} \hat{\mathbf{s}}_i[k+1],\\
\hat{\mathbf{z}}_{\mathcal{UO}_i}[k+1]&=\bar{\mathbf{J}}_{\mathcal{UO}_i}\mathbb{B}_i \sum_{l\in\mathcal{N}_i}\mathbf{w}_{il} \otimes ({\mathcal{T}}^{-1}\hat{\mathbf{x}}_{l}[k]),\\
\hat{\mathbf{x}}_i[k+1]&=\mathcal{T}\mathbb{P}_i\begin{bmatrix} \hat{\mathbf{z}}_{\mathcal{O}_i}[k+1] \\ \hat{\mathbf{z}}_{\mathcal{UO}_i}[k+1] \end{bmatrix}.
\end{empheq}
\label{eqn:Cond2obs}
\end{subequations}
Note that since $\mathcal{T}$ depends only on the system matrix $\mathbf{A}$ which is assumed to be time-invariant, the term $\mathcal{T}^{-1}$ appearing in (\ref{eqn:Cond2obs}c) needs to be computed only once.
\subsection{Analysis of the Estimation Scheme for Systems and Graphs Satisfying Condition $2$}
The following is the main result of this section.

\begin{theorem}
\label{theo:L2}
Consider a system $(\mathbf{A,C})$ and graph $\mathcal{G}$ satisfying Condition $2$. Then, for each node $i \in \{1,2, \cdots, N\}$, there exists a choice of observer gain matrix $\mathbb{L}_i$, and consensus weights $w^{j}_{il}$, $\forall \lambda_j \in \mathcal{UO}_i$, $l\in\mathcal{N}_i$, such that the update rules given by (\ref{eqn:Cond2obs}) form a distributed observer.
\end{theorem}

\begin{proof}
Let a system $(\mathbf{A,C})$ and graph $\mathcal{G}$ satisfy Condition $2$. Consider $\lambda_j \in \Lambda_{U}(\mathbf{A}).$ Let $\mathcal{S}_j=\{m_1, \cdots, m_{\tau_j}\}$ be the set of root nodes for eigenvalue $\lambda_j$, where $\tau_j=|\mathcal{S}_j|$. Define $\mathbf{e}^{(j)}_{m_i}[k] \triangleq \hat{\mathbf{z}}^{(j)}_{m_i}[k]-\mathbf{z}^{(j)}[k]$ as the error in estimation of the $j$-th sub-state by the $m_i$-th node. The errors in estimation of $\mathbf{z}^{(j)}[k]$ for the nodes that can detect $\lambda_j$ are stacked into the composite error vector $\mathbf{E}^{(j)}_{\mathcal{O}}[k]$, defined as 
\begin{equation}
\mathbf{E}^{(j)}_{\mathcal{O}}[k]\triangleq \left[\begin{array}{c}\mathbf{e}^{(j)}_{m_1}[k]\\ \vdots \\ \mathbf{e}^{(j)}_{m_{\tau_j}}[k]\end{array}\right].
\end{equation}
 Similarly, we stack the estimation errors of $\mathbf{z}^{(j)}[k]$ for the nodes that cannot detect $\lambda_j$ into the composite error vector $\mathbf{E}^{(j)}_{\mathcal{UO}}[k]$, defined as
 \begin{equation}
\mathbf{E}^{(j)}_{\mathcal{UO}}[k]\triangleq \left[\begin{array}{c}\mathbf{e}^{(j)}_{m_\mathsmaller{{\tau_j+1}}}[k]\\ \vdots \\ \mathbf{e}^{(j)}_{m_{N}}[k]\end{array}\right],
\end{equation}
where $\mathcal{V}\setminus\mathcal{S}_j=\{m_{\mathsmaller{\tau_j+1}}, \cdots, m_{N}\}$ represents a topological ordering of the non-root nodes consistent with a set of directed trees rooted at $\mathcal{S}_j$, which span $\mathcal{V}\setminus\mathcal{S}_j$. Such a set of trees exists based on Condition 2. Noting from (\ref{eqn:REPsystem}) that $\mathbf{z}^{(j)}[k+1]=\mathbf{J}_j\mathbf{z}^{(j)}[k]$, and using equation (\ref{eqn:consensusREP}), for $\lambda_j \in \mathcal{UO}_i$, the estimation error for the $j$-th sub-state by the $i$-th node is
\begin{equation}
\mathbf{e}^{(j)}_{i}[k+1]=\mathbf{J}_j\sum_{l\in\mathcal{N}_i}w^{j}_{il}\mathbf{e}^{(j)}_{l}[k].
\label{eqn:indverrcons}
\end{equation}
From (\ref{eqn:indverrcons}), it follows that the relation between $\mathbf{E}^{(j)}_{\mathcal{O}}[k]$ and $\mathbf{E}^{(j)}_{\mathcal{UO}}[k]$ can be expressed via the equation
\begin{equation}
\begin{split}
\mathbf{E}^{(j)}_{\mathcal{UO}}[k+1]&=\left(\left[\begin{array}{c c}\mathbf{W}^{j}_{11}&\mathbf{W}^{j}_{12}\end{array}\right]\otimes\mathbf{J}_j\right)\left[\begin{array}{c}\mathbf{E}^{(j)}_{\mathcal{O}}[k] \\ \mathbf{E}^{(j)}_{\mathcal{UO}}[k] \end{array}\right]\\
&=\left(\mathbf{W}^{j}_{12}\otimes\mathbf{J}_j\right)\mathbf{E}^{(j)}_{\mathcal{UO}}[k]+\left(\mathbf{W}^{j}_{11}\otimes\mathbf{J}_j\right)\mathbf{E}^{(j)}_{\mathcal{O}}[k],
\end{split}
\label{eqn:ErrcompREP}
\end{equation}
where the weight matrix $\mathbf{W}^{j}=\left[\begin{array}{c c}\mathbf{W}^{j}_{11}&\mathbf{W}^{j}_{12}\end{array}\right]$ contains weights based on equation (\ref{eqn:indverrcons}) (note that $\mathbf{W}^{j}\in \mathbb{R}^{(N-\tau_j) \times N}$, and $\mathbf{W}^{j}_{11}$ represents the first $\tau_j$ columns of $\mathbf{W}^{j}$, where $\tau_j=|\mathcal{S}_j|$). Using the same design philosophy for the consensus weights as in Condition 1, we assign non-zero consensus weights  only along the branches of the spanning forest rooted at $\mathcal{S}_j$. In this way, $\mathbf{W}^{j}_{12}$ can be made lower triangular with zero eigenvalues (without violating the stochasticity condition imposed by equation (\ref{eqn:REPstochasticity})). We conclude that by an appropriate choice of weights as described above, we can achieve $\Lambda_{U}(\mathbf{W}^{j}_{12} \otimes \mathbf{J}_j) = \emptyset$ (even though $\lambda_j \in \Lambda_{U}(\mathbf{A}))$.

Based on Lemma \ref{lemmaSEP}, each node $i \in \mathcal{V}$ can locally design its observer gain $\mathbb{L}_i$ to stabilize the local Luenberger observer error dynamics. Specifically, the error dynamics corresponding to the estimation of $\mathbf{z}^{(j)}[k]$ for each root node in $\mathcal{S}_j$ is guaranteed to asymptotically converge based on Lemma \ref{lemmaSEP}, i.e., the composite error $\mathbf{E}^{(j)}_{\mathcal{O}}[k]$ asymptotically converges to zero. Using ISS, we infer from (\ref{eqn:ErrcompREP}) that $\lim_{k\to\infty} \mathbf{E}^{(j)}_{\mathcal{UO}}[k]=\mathbf{0}$. The same argument holds $\forall \lambda_j \in \Lambda_{U}(\mathbf{A})$. Thus, we conclude that every node can asymptotically estimate the sub-states of $\mathbf{z}[k]$ corresponding to both its detectable and undetectable eigenvalues, i.e., it can estimate the entire transformed state $\mathbf{z}[k]$ asymptotically. As $\mathbf{x}[k]=\mathcal{T}\mathbf{z}[k]$, each node can asymptotically estimate the true state $\mathbf{x}[k]$ as well.
\end{proof}
\begin{remark}
Based on the distributed observer given by (\ref{eqn:Cond2obs}), note that the dimension of the observer is equal to the sum of the dimensions of the vectors $\hat{\mathbf{s}}_i[k]$ and $\hat{\mathbf{z}}_{\mathcal{UO}_i}[k]$, and can be higher than the dimension of the state $\mathbf{x}[k]$ (as $\hat{\mathbf{s}}_i[k]$ can have a dimension higher than $\hat{\mathbf{z}}_{\mathcal{O}_i}[k]$). This augmentation is a consequence of the fact that at present, although we are able to estimate the portion $\mathbf{w}_{\mathcal{O}_i}[k]$ of the vector $\mathbf{z}_{\mathcal{UO}_i}[k]$ via the local Luenberger observer maintained at node $i$ (given by (\ref{eqn:luenCase1})), we use this information only for updating $\hat{\mathbf{z}}_{\mathcal{O}_i}[k]$, and rely on consensus for estimating the entire vector $\hat{\mathbf{z}}_{\mathcal{UO}_i}[k]$ (via equation (\ref{eqn:Cond2obs}c)). This redundancy in information may be potentially overcome using a more complicated scheme where one uses consensus for estimating only the portion of the state corresponding to the vector $\mathbf{w}_{\mathcal{UO}_i}[k]$ in equation (\ref{eqn:Case1obs1}); to avoid cluttering the exposition, we omit further investigation of this issue in the present paper. However, for certain special cases of Condition $2$ where the system matrix has more structure, it may be possible to construct distributed  observers without state augmentation, using the approach proposed for Condition $2$. For example, if the system has distinct eigenvalues, then the matrix $\bar{\mathbf{C}}_{\mathcal{UO}_i}$ in \eqref{eqn:Case1form} will be zero $\forall i\in\mathcal{V}$, thereby precluding the need for state augmentation.
\end{remark}
\subsection{Summary of the Estimation Scheme for Systems and Graphs Satisfying Condition $2$}
Similar to the strategy adopted for systems and graphs satisfying Condition 1, the distributed observer design for systems and graphs satisfying Condition 2 also constitutes an initialization (or design) phase which needs to be implemented just once, followed up by an estimation phase. However, the extra structure provided by Condition 2 allows each of these phases to be implemented in a distributed manner. The main steps of the overall scheme are summarized as follows:

\textbf{Design Phase:}
\begin{itemize}
\item All nodes simultaneously perform a common co-ordinate transformation, which brings the state matrix $\mathbf{A}$ into its Jordan canonical form. Using this form, each node identifies its locally detectable and undetectable eigenvalues.
\item For each $\lambda_j \in \Lambda_{U}(\mathbf{A})$, the nodes run a distributed algorithm to construct trees with roots in $\mathcal{S}_j$, which span $\mathcal{V}\setminus\mathcal{S}_j.$\footnote{We shall discuss such an algorithm shortly.} These trees guide the design of the consensus weight matrices to be used for each unstable and marginally stable eigenvalue of the system.
\end{itemize}

\textbf{Estimation Phase:}
\begin{itemize}
\item Each node uses a Luenberger observer for estimating the sub-states corresponding to the detectable eigenvalues, and runs consensus dynamics for estimating the sub-states corresponding to the undetectable eigenvalues. These dynamics are captured by (\ref{eqn:Cond2obs}).
\end{itemize}

\subsubsection*{Construction of Spanning Trees for Consensus Weight Design} 

To construct directed trees rooted at nodes in $\mathcal{S}_j$, which span $\mathcal{V}\setminus\mathcal{S}_j$, for each $\lambda_j \in \Lambda_{U}(\mathbf{A})$, (these trees in turn guide the construction of the consensus weight matrices) one can use standard distributed tree construction algorithms (such as Breadth-First Search (BFS)) \cite{lynch}. The essential idea behind such algorithms is that each desired root node broadcasts a message indicating that it is a root, which is then passed through the network.  When a node first receives such a message from a neighbor, it adopts that neighbor as its parent in the tree and rebroadcasts the message.  At the conclusion of the algorithm, all nodes are aware of their parent in their tree (as long as there is a path from the root node(s) to all other nodes).  
For our purpose, such a distributed algorithm can be implemented by the nodes for each $\lambda_j\in\Lambda_{U}(\mathbf{A})$, with $\mathcal{S}_j$ representing the roots of the tree. In this way, for each $\lambda_j\in\Lambda_{U}(\mathbf{A})$, a node in $\mathcal{V}\setminus\mathcal{S}_j$ will identify its parent node in one of the directed trees rooted at $\mathcal{S}_j$, and as discussed in the proof of Theorem \ref{theo:L2}, will assign a non-zero consensus weight to only this parent node for the estimation of $\mathbf{z}^{(j)}[k]$.  

Note that the simpler distributed observer scheme developed for systems and graphs satisfying Condition 2 may not always be applicable to systems and graphs satisfying Condition 1. To see this, consider the system and graph given by equation (\ref{eqn:examplesystem}) and Figure \ref{fig:example}, which satisfies Condition 1 but not Condition 2. As pointed out in Remark \ref{remark:Cond}, the only root node that can detect the unstable eigenvalue $\lambda=2$ belongs to the source component comprised of the isolated node 3. To implement the scheme developed for systems and graphs satisfying Condition 2, one needs to construct a tree rooted at node 3 which spans nodes 1 and 2. This is clearly not possible for this particular network; hence the method developed for Condition 2 is not applicable to this system and graph. In this case, the general distributed observer framework developed for systems and graphs satisfying Condition 1 would still apply, however.

\section{Robustness to Communication Losses and Sensor Failures}
\label{sec:timevar}
The general framework for distributed observer design that we have described thus far (the idea of using Luenberger observers for estimating the locally detectable states and consensus dynamics for the remaining states) allows for various extensions. Here, we discuss how our approach can deal with communication losses and sensor failures or attacks.
\subsection{Extension to Communication Losses} We first discuss how the estimation scheme developed for systems and graphs satisfying Condition 1 can be extended to account for  time-varying communication graphs that are a consequence of communication link failures.\footnote{A similar analysis can be performed for Condition 2 and hence the details are omitted for brevity.} In other words, the network varies with time due to failure or recovery of subsets of edges of the baseline graph $\mathcal{G}$. We denote this class of switching signals by $\Omega$. Under the class of switching signals  $\Omega$, the time-varying communication graph is denoted by $\mathcal{G}_{\sigma(k)}=(\mathcal{V,E}_{\sigma(k)})$, where $\sigma(k)$ is a finite index set representing the different switching modes, and $\mathcal{E}_{\sigma(k)}\subseteq\mathcal{E}$ (recall that $\mathcal{E}$ represents the set of edges of the baseline graph $\mathcal{G}$). For the rest of the analysis in this section, we assume  that the baseline communication graph $\mathcal{G}$ is strongly-connected, i.e., $\mathcal{G}$ is strongly-connected in the absence of link failures.

We make two minor modifications to our original estimation strategy (refer to Section \ref{sec:summary1}) to account for communication losses. First, during the design phase, for estimation of a given sub-state, we construct a spanning directed acyclic graph (DAG) (instead of a spanning tree) rooted at the corresponding source node to allow for the possibility of having redundant communication links. Accordingly, in the DAG constructed for  estimation of sub-state $j$ (where $j\in\{1,\cdots,N\}$), let the set of parent nodes for node $i$ be denoted by $\mathcal{P}^{(j)}_{i}$. Second, for  estimating the $j$-th sub-state, where $j \in \{1, \cdots ,N\}\setminus\{i\}$, the $i$-th node does the following: if at a certain time-step it receives information from only a proper subset of its parent set $\mathcal{P}^{(j)}_{i}$, then it still employs equation \eqref{eqn:consensus}, redistributing the weights among such a subset so as to preserve the stochasticity constraint imposed by \eqref{eqn:stochasticity}.\footnote{Notice that for a given node $i$, the observer update equations \eqref{eqn:luen} and \eqref{eqn:unobsestimate} are unaffected by changes in the network structure.} For the more critical scenario where node $i$ gets disconnected from all its parents in the set $\mathcal{P}^{(j)}_{i}$ at a given time-step $k$, it updates $\hat{\mathbf{z}}^{(j)}_{i}[k]$ using previous values of its \textit{own} estimates in the following way: 
\begin{equation}
\hat{\mathbf{z}}^{(j)}_{i}[k+1]=\mathbf{A}_{jj}\hat{\mathbf{z}}^{(j)}_{i}[k]+\sum_{l=1}^{j-1}\mathbf{A}_{jl}\hat{\mathbf{z}}^{(l)}_{i}[k].
\label{eqn:consensus_timevar}
\end{equation}

A direct consequence of the update rule \eqref{eqn:consensus_timevar} is that the matrix $\mathbf{A}_{jj}$, which may be unstable, appears in the block diagonal position corresponding to node $i$ in the lower block triangular error dynamics matrix $\mathbf{M}_j$ given by \eqref{eqn:bigerrmatrix}. 
We make the following assumption on the class of switching signals $\Omega$.
\begin{assumption} $\Omega$ has the following property: there exists a positive integer $T$ such that in every time interval of the form $[kT,(k+1)T]$, where $k\in\mathbb{N}$, for each sub-state $j\in\{1,\cdots,N\}$, for every node $i\in\mathcal{V}\setminus\{j\}$, there exists an integer $l \in [kT,(k+1)T]$ such that the graph $\mathcal{G}_{\sigma{(l)}}$ contains an edge from at least one node in $\mathcal{P}^{(j)}_i$ to node $i$.\footnote{Note that within a given interval of the form $[kT,(k+1)T]$, for a specific sub-state $j$, $l$ might be different for the different non-source nodes in $\mathcal{V}\setminus\{j\}$.}
\label{assump:switch}
\end{assumption}
In words, Assumption \ref{assump:switch} simply implies that within each interval $[kT,(k+1)T]$, for each sub-state $j$, every non-source node $\mathcal{V}\setminus\{j\}$ is guaranteed to receive information from at least one of its parents in $\mathcal{P}^{(j)}_i$ at least once over the entire interval. With this in mind, we now state the main result of this section.
\begin{theorem}
\label{theo:timevar2}
Consider a system $(\mathbf{A,C})$ and strongly-connected baseline communication graph $\mathcal{G}$  satisfying Condition 1. Suppose the class of switching signals $\Omega$ satisfies Assumption \ref{assump:switch}. Then, equations (\ref{eqn:luen}), (\ref{eqn:consensus}) with time-varying weights,  (\ref{eqn:unobsestimate}) and \eqref{eqn:consensus_timevar} form a distributed observer.
\end{theorem}

The proof is provided in Appendix \ref{sec:proof_timevarconstr}. 
\subsection{Extension to Node Attacks}
Dealing with node failures (or more generally, malicious node behavior) requires additional careful analysis and incorporation of redundancy in the network and sensors so as to meet a two-fold objective - the collective measurements of the non-compromised nodes must ensure global detectability of the state, and the network must contain a sufficient amount of redundancy to prevent the failed (or malicious) nodes from acting as bottlenecks between the correctly functioning nodes. Furthermore, even with such redundancy, one has to carefully design the estimation algorithm to leverage the redundancy so as to guarantee asymptotic state reconstruction by the non-faulty nodes. In addition to these challenges, one needs to account for the fact that under malicious attacks, a node might receive corrupt information from compromised neighboring nodes. This, in general, is a non-trivial task. A preliminary investigation of the node attack scenario (which subsumes sensor failures) has been conducted in our recent work \cite{mitraCDC} where we provide sufficient conditions on the network that guarantee asymptotic state reconstruction for systems with distinct real eigenvalues (a special case of Condition $2$) despite a certain number of malicious nodes in the network.\footnote{In addition to admitting a fully distributed implementation as mentioned earlier, robustness to node attacks is another key advantage afforded by our approach for systems and graphs satisfying Condition 2. This is in contrast with existing work on distributed observers, which, to the best of our knowledge, requires a centralized design phase and does not provide any theoretical guarantees against node faults or attacks.} Specifically, we show that under the provided conditions on the network topology, we can extend the approach that we have developed in this paper to construct robust DAGs, along with a slightly modified consensus update rule, in order to provide security guarantees for the normal nodes.  This illustrates the benefits of the framework for distributed estimation that we have provided in this paper; extensions of these ideas to tackle malicious behavior under more general system dynamics is an area of ongoing research.

\begin{figure}[t]
\begin{center}
\begin{tikzpicture}
[->,shorten >=1pt,scale=.5, minimum size=10pt, auto=center, node distance=3cm,
  thick, every node/.style={circle, draw=black, thick},]
\tikzstyle{block} = [rectangle, draw, fill=green!20, 
    text width=5em, text centered, rounded corners, minimum height=4em];
    \node [block]  at (3,4) (plant) {LTI plant};
\node [circle, draw, fill=blue!20](n1)  at (0,0)  (1) {1};
\node [circle, draw, fill=blue!20](n2) at (3,0)   (2)  {2};
\node [circle, draw, fill=blue!20](n3) at (6,-0)  (3)  {3};

\path[every node/.style={font=\sffamily\small}]
(plant)
             edge[red, bend right] node [right] {$y_1[k]$} (1)
             edge[red] node [right] {$y_2[k]$} (2)
                       
    (1)
        edge [bend right] node [] {} (2)
        
    (2)
        edge [bend right] node [] {} (1)
        edge [bend left] node [] {} (3);

\end{tikzpicture}
\end{center}
\caption{Network topology for the example system considered in Section~\ref{sec:example}.}
\label{fig:sim}
\end{figure}
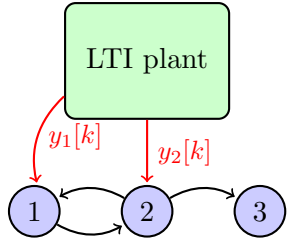

\section{Example}
\label{sec:example}
\begin{figure}
\begin{center}
\begin{tabular}{c c}
\includegraphics[scale=0.4]{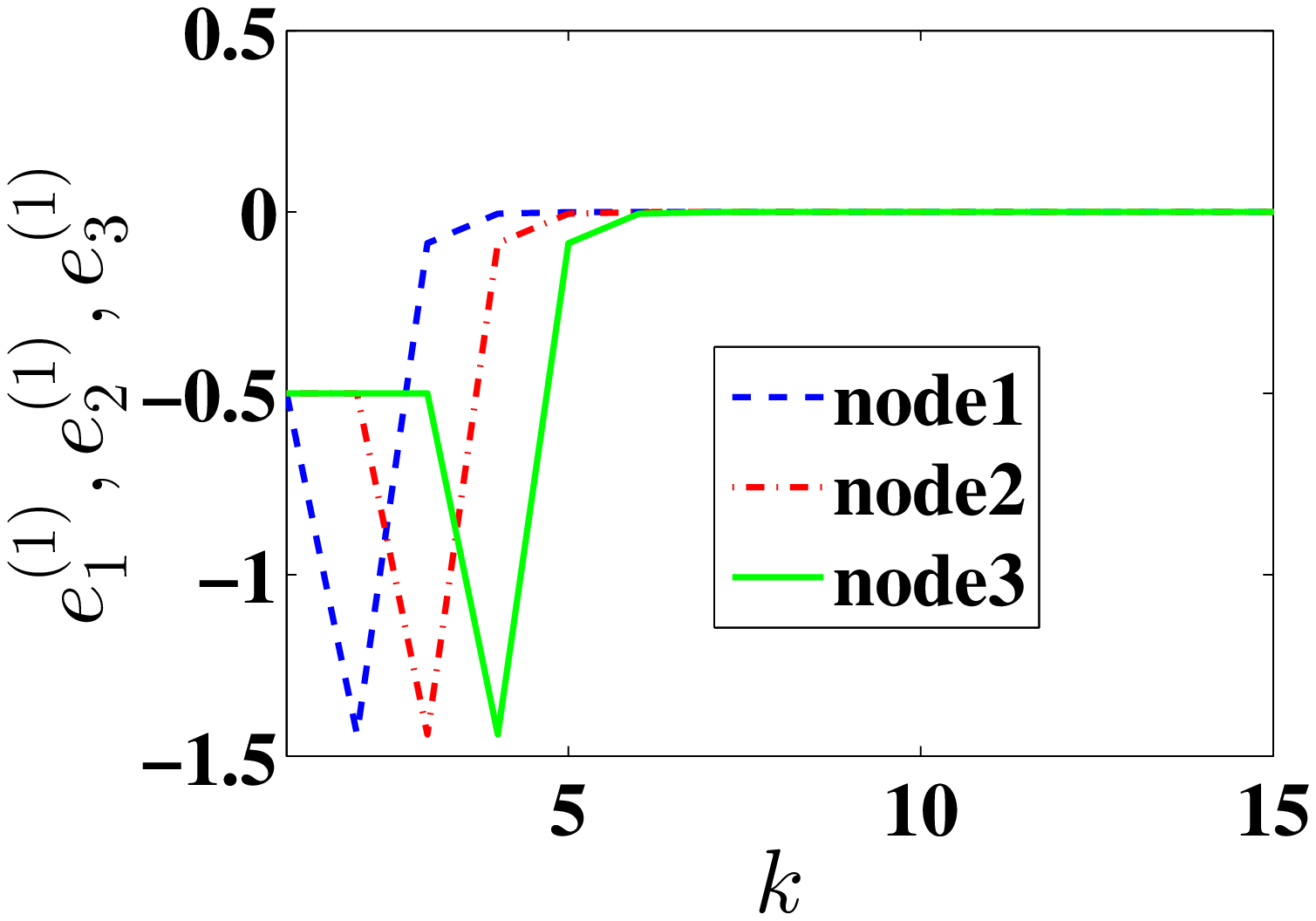}  & \includegraphics[scale=0.4]{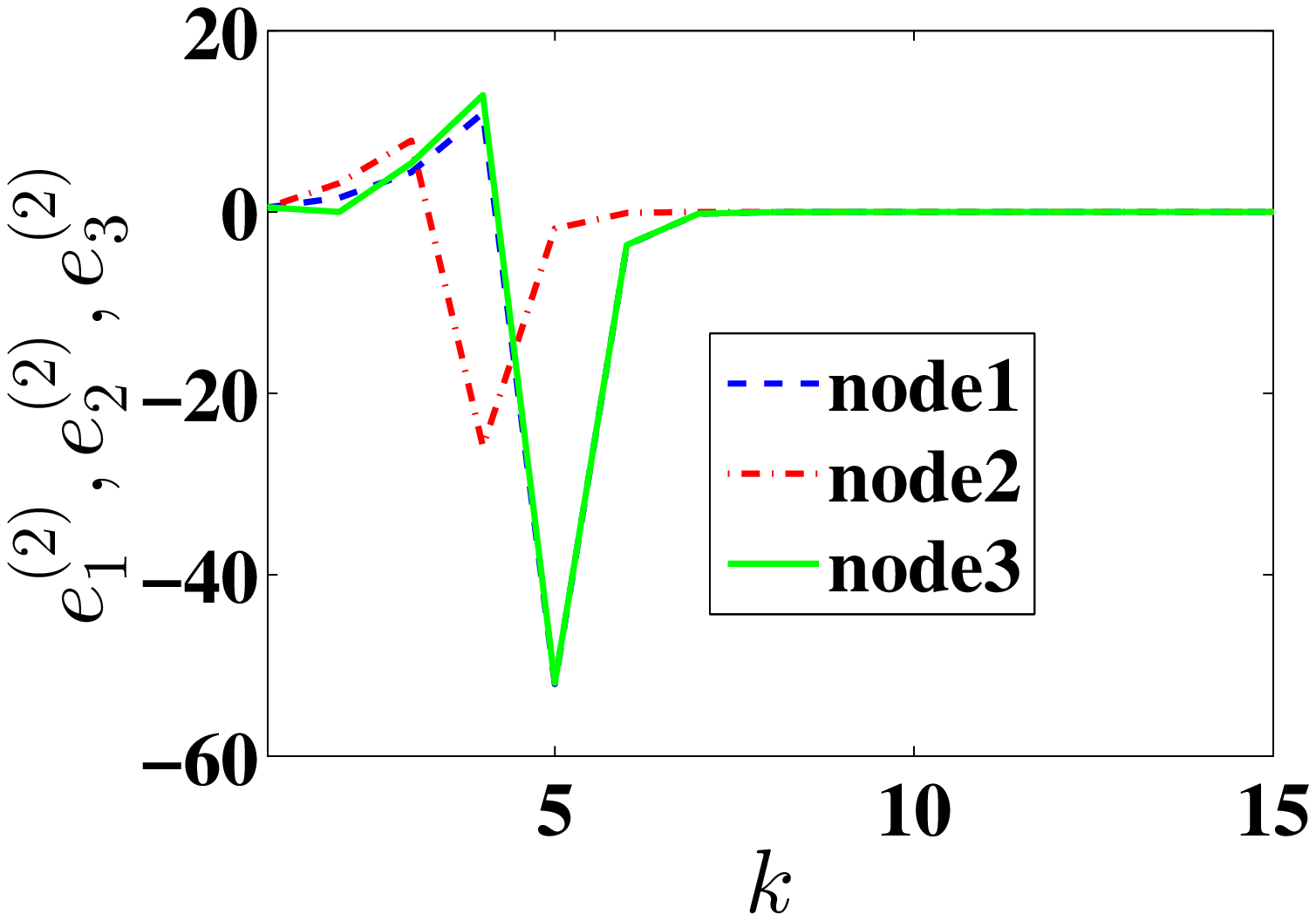}\\ 
\includegraphics[scale=0.4]{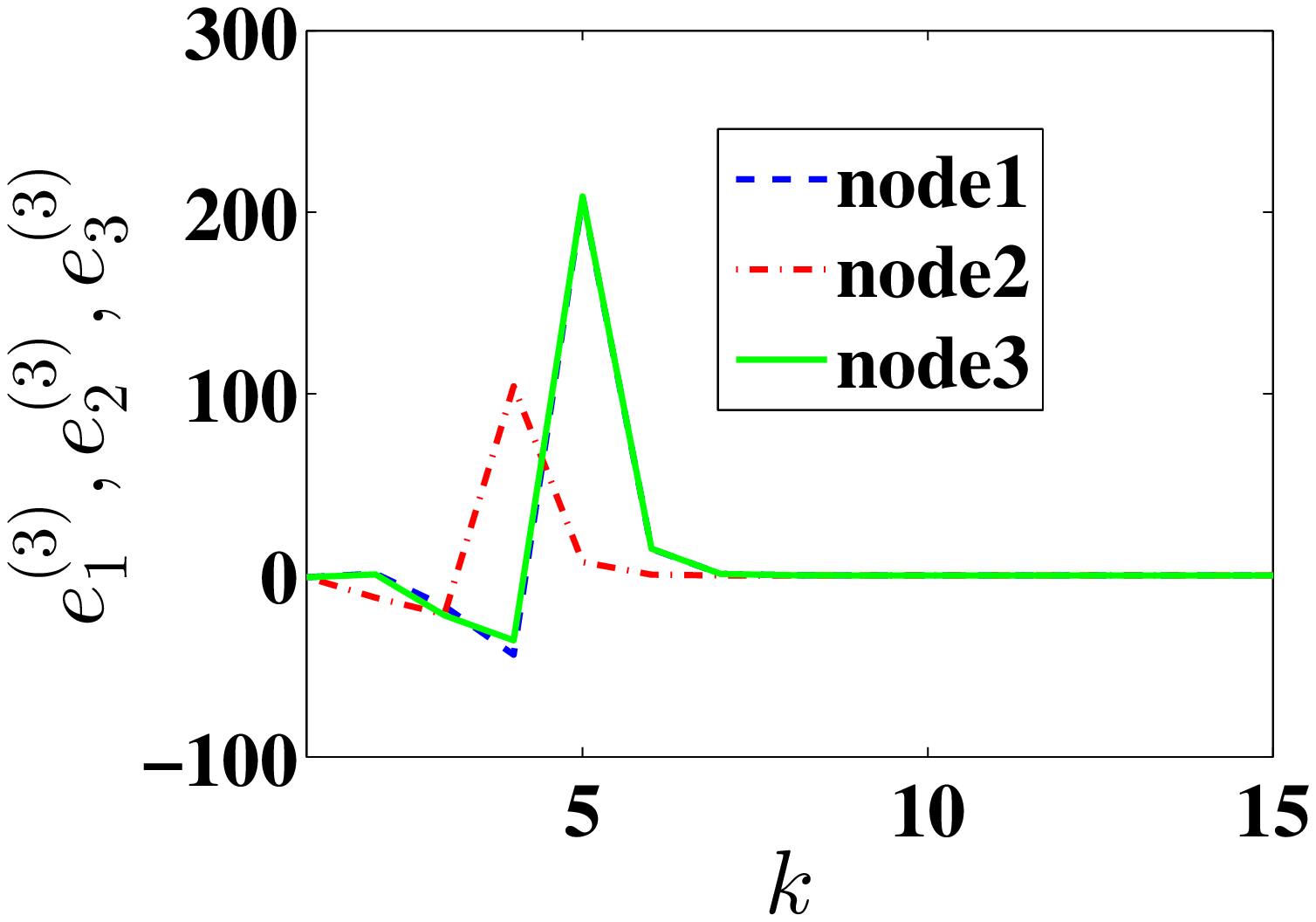}  &
\end{tabular}
\caption{{\it (Top left)} Error dynamics for the first state. {\it (Top right)} Error dynamics for the second state. {\it (Bottom left)} Error dynamics for the third state.} 
\label{fig:Simexamples}
\end{center}
\end{figure}

In this section, we present an example to illustrate the scheme developed for Condition 1. To this end, consider the network shown in Figure \ref{fig:sim}, and the associated system and measurement matrices given by
\begin{equation}
\mathbf{A}=\begin{bmatrix} 1&0&0\\2&2&0\\-5&0&2 \end{bmatrix}, \mathbf{C}_1=\begin{bmatrix} 4&4&1 \end{bmatrix}, \mathbf{C}_2=\begin{bmatrix} 11&13&3\\16&18&4 \end{bmatrix}, \mathbf{C}_3=0.
\label{eqn:Sim_example}
\end{equation}
Note that the system is not detectable from any individual node. It can be verified that the pair $(\mathbf{A},\begin{bmatrix} \mathbf{C}_1 \\ \mathbf{C}_2 \end{bmatrix})$, associated with the source component comprised of nodes 1 and 2, is observable. Hence, the scheme developed for Condition 1 is applicable to this setting. To implement the multi-sensor observable canonical decomposition, we start by bringing the pair $(\mathbf{A},\mathbf{C}_1)$ to the observable canonical form.  
This is achieved using $\mathbf{T}_1 = \left[\begin{smallmatrix} 4&7&0\\4&8&-0.2425\\1&2&0.9701 \end{smallmatrix}\right]$. Since $(\mathbf{A},\begin{bmatrix} \mathbf{C}_1 \\ \mathbf{C}_2 \end{bmatrix})$ is observable, in this specific case, we have $\mathbf{A}_{22}=\mathbf{A}_{1\mathcal{U}}$ and $\mathbf{A}_{2\mathcal{U}}=0$ (here we use notations consistent with the ones described in the proof of Proposition \ref{transformations} in Appendix  \ref{sec:proof_multisensor}). Thus, we have $\mathbf{T}_2=\mathbf{I}_3$ and $\mathcal{T}=\mathbf{T}_1\mathbf{T}_2=\mathbf{T}_1$. Using $\mathcal{T}$, we perform the multi-sensor observable canonical decomposition to obtain
\begin{equation}
\bar{\mathbf{A}}=\begin{pmat}[{.|}]
-10.9412 & -22.6471 & 0 \cr
6.8235 & 13.9412 & 0 \cr\-
-21.3431 & -37.3505 & 2 \cr
\end{pmat}, \hspace{1mm} \bar{\mathbf{C}}_1=\begin{pmat}[{.|}] 33&62&0 \cr \end{pmat}, \hspace{1mm}
\bar{\mathbf{C}}_2=\begin{pmat}[{.|}] 99&187&-0.2425\cr
                                      140&264&-0.4851\cr
                                      \end{pmat}.
\label{eqn:Sim_multi}
\end{equation}
Based on the above decomposition, and the theory developed for Condition 1, it is easy to see that the transformed state $\mathbf{z}[k]=\mathcal{T}^{-1}\mathbf{x}[k]$ will contain two sub-states and the unobservable sub-state will be of zero-dimension. The local Luenberger observer gains are chosen as $\mathbf{L}_1={\begin{bmatrix} -4.6404 & 2.5174 \end{bmatrix}}^{T}$ and $\mathbf{L}_2=\begin{bmatrix} -1.641 & -3.282 \end{bmatrix}$. Noting that node 1 is responsible for estimating sub-state 1, and node 2 for sub-state 2, the consensus weight vectors used by the two nodes are $\mathbf{w}_{11}={\begin{bmatrix} 1&0 \end{bmatrix}}^{T}, \mathbf{w}_{12}={\begin{bmatrix} 0&1 \end{bmatrix}}^{T}, \mathbf{w}_{21}={\begin{bmatrix} 1&0 \end{bmatrix}}^{T}$, and $\mathbf{w}_{22}={\begin{bmatrix} 0&1 \end{bmatrix}}^{T}$.\footnote{The reader is referred to the discussion in Section \ref{sec:compact} for a description of these weight vectors.} Using these design parameters, nodes 1 and 2 maintain the following estimators of the form (\ref{eqn:Overallest}):
\begin{align}
\hat{\mathbf{x}}_1[k+1]&=\mathbf{N}\hat{\mathbf{x}}_1[k]+\mathcal{T}\mathbb{H}_1\left(\mathbf{y}_1[k]-\mathbf{C}_1\hat{\mathbf{x}}_1[k]\right)+\mathbb{G}_{11}\hat{\mathbf{x}}_1[k]+\mathbb{G}_{12}\hat{\mathbf{x}}_2[k],\nonumber\\
\hat{\mathbf{x}}_2[k+1]&=\mathbf{N}\hat{\mathbf{x}}_2[k]+\mathcal{T}\mathbb{H}_2\left(\mathbf{y}_2[k]-\mathbf{C}_2\hat{\mathbf{x}}_2[k]\right)+\mathbb{G}_{21}\hat{\mathbf{x}}_1[k]+\mathbb{G}_{22}\hat{\mathbf{x}}_2[k],\nonumber
\end{align}
where
\begin{equation}
\mathbf{N}=\begin{bmatrix}0&0&0\\1.29&0&0\\-5.18&0&0\end{bmatrix}, \mathcal{T}\mathbb{H}_1=\begin{bmatrix}-0.94\\1.58\\0.39\end{bmatrix},\mathcal{T}\mathbb{H}_2=\begin{bmatrix}0&0\\0.40&0.80\\-1.59&-3.18\end{bmatrix},\nonumber
\end{equation}
\begin{equation}
\mathbb{G}_{11}=\mathbb{G}_{21}=\begin{bmatrix}1&0&0\\0.71&1.88&0.47\\0.18&0.47&0.12\end{bmatrix}, \mathbb{G}_{12}=\mathbb{G}_{22}=\begin{bmatrix}0&0&0\\0&0.12&-0.47\\0&-0.47&1.88\end{bmatrix}.\nonumber
\end{equation}
Since node 3 does not belong to any source component, it simply runs a pure consensus strategy given by $\hat{\mathbf{x}}_3[k+1]=\mathbf{A}\hat{\mathbf{x}}_2[k]$. Note that all nodes maintain observers of dimension $3$. For simulations, we let the true state $\mathbf{x}[k]$ evolve from the initial condition $\mathbf{x}[0]={\begin{bmatrix} 0.5&-0.5&1 \end{bmatrix}}^{T}$, while the initial estimates of all three nodes are set to zero. Figure \ref{fig:Simexamples} shows the estimation errors associated with each of the three states of the system (in these plots, the notation $e^{(j)}_i$ is used to denote the error in estimation of state $j$ by node $i$). These plots validate the scheme developed for Condition 1.
 
 
\section{Conclusions and Future Work}
\label{sec:conclusion}
In this paper, we considered the problem of distributed state estimation of a LTI system by a network of nodes. We introduced a new class of distributed observers (of the form (\ref{eqn:Overallest})), for the most general class of LTI systems, directed communication graphs and linear sensor measurement structures. This was achieved by extending the Kalman observable canonical decomposition to a setting with multiple sensors, i.e., by introducing the notion of a multi-sensor observable canonical decomposition. We also demonstrated that for certain subclasses of system dynamics and networks, one can design a distributed observer via a simpler estimation scheme which enjoys the benefit of a fully distributed design phase. The main underlying theme of our work, which unifies each of the two estimation schemes presented in this paper, is built upon the following intuition: portions of the state space can be reconstructed by a node using its own local measurements, and hence it needs to run consensus for only the portion of the state space that is not locally detectable. Finally, we discussed how our proposed framework can be extended to account for certain classes of time-varying networks as well. 

As future work, it would be interesting to explore whether low-complexity distributed algorithms for observer design (such as the one we presented for Condition 2) can be obtained for systems and graphs satisfying Condition 1. Furthermore, while we focused on noiseless dynamics in this paper for ease of exposition (like \cite{Khanobs1,Khanobs2,wang,martins3 ,ugrinov}), for stochastic systems, one needs to determine the optimal way of combining the information received from different sensors while designing consensus protocols.\footnote{It can be shown that the methods developed in this paper lead to bounded mean square estimation error in the presence of i.i.d. process and measurement noise with bounded second moments.} Another important problem is the design of distributed functional observers motivated by the following observation: for large systems, it may be computationally demanding for  every node to estimate the entire state vector; in practice, the nodes may only care about estimating certain linear functions of the state. Finally, an investigation of how the proposed framework fares against network induced issues   such as delays and communication asynchronicities would also be interesting.

\bibliographystyle{unsrt} 
\bibliography{References}

\appendix

\section{Proof of Proposition \ref{transformations}}
\label{sec:proof_multisensor}

\begin{proof}
We outline the sequence of transformations that need to be carried out.

\textbf{Step 1} : \textit{Transformation at Sensor 1}

Consider the coordinate transformation $\mathbf{x}[k]=\mathbf{T}_1\mathbf{z}_1[k]$. Here, $\mathbf{T}_1$ is the non-singular matrix that performs an observable canonical decomposition of the pair $(\mathbf{A},\mathbf{C}_1)$, yielding
\begin{equation}
\begin{split}
\underbrace{\left[\begin{array}{c}
\mathbf{z}^{(1)}_{}[k+1]\\
\mathbf{z}^{(1)}_{\mathcal{U}}[k+1]
\end{array}\right]}_{\mathbf{z}_1[k+1]} &= \underbrace{\begin{bmatrix} 
                      \mathbf{A}_{11} & \mathbf{0} \\
                      \mathbf{A}_{1X}          & \mathbf{A}_{1\mathcal{U}}
  \end{bmatrix} }_{\bar{\mathbf{A}}_1=\mathbf{T}^{-1}_1\mathbf{A}\mathbf{T}_1}                    
\underbrace{\left[\begin{array}{c}
\mathbf{z}^{(1)}_{}[k]\\
\mathbf{z}^{(1)}_{\mathcal{U}}[k]
\end{array}\right]}_{\mathbf{z}_1[k]}, \\
\mathbf{y}_1[k] &= \underbrace{\begin{bmatrix} \mathbf{C}_{11} & \mathbf{0} \end{bmatrix}}_{\bar{\mathbf{C}}_1 = \mathbf{C}_1\mathbf{T}_1}\mathbf{z}_1[k]. \enspace 
\end{split}
\label{eqn:obs1}
\end{equation}
Thus, we have
\begin{equation}
\begin{split}
\mathbf{z}^{(1)}_{}[k+1]&=\mathbf{A}_{11}\mathbf{z}^{(1)}_{}[k],\\
\mathbf{y}_1[k]&= \mathbf{C}_{11}\mathbf{z}^{(1)}_{}[k].
\end{split}
\label{eqn:truedyn1}
\end{equation}
Let $\mathbf{z}^{(1)}_{}[k] \in {\mathbb{R}}^{o_1}$. 

\textbf{Step 2} : \textit{Transformation at Sensor 2}

We know the following:
\begin{equation}
\begin{split}
\mathbf{y}_2[k]&=\mathbf{C}_2\mathbf{x}[k]=\mathbf{C}_2\mathbf{T}_1\mathbf{z}_1[k]\\
       &\triangleq\begin{bmatrix} \mathbf{C}_{21} & \mathbf{C}_{21\mathcal{U}} \end{bmatrix} \left[\begin{array}{c}
\mathbf{z}^{(1)}_{}[k]\\
\mathbf{z}^{(1)}_{\mathcal{U}}[k]
\end{array}\right].
 \end{split}
 \end{equation}
 Let $\bar{\mathbf{T}}_2$ be a non-singular transformation matrix that performs an observable canonical decomposition of the pair $(\mathbf{A}_{1\mathcal{U}},  \mathbf{C}_{21\mathcal{U}})$. We now wish to identify the portion of the unobservable subspace of sensor $1$ that is observable with respect to sensor $2$. With this objective in mind, consider the coordinate transformation $\mathbf{z}_1[k]=\mathbf{T}_2\mathbf{z}_2[k]$, where the non-singular transformation matrix $\mathbf{T}_2$ is defined as
\begin{equation}
\mathbf{T}_2=\begin{bmatrix} \mathbf{I}_{o_1} & \mathbf{0} \\ \mathbf{0} & \bar{\mathbf{T}}_2
 \end{bmatrix}.
\label{eqn:T2}
\end{equation}
 This yields the following dynamics:
 \begin{equation} 
 \begin{split}
 \underbrace{\left[\begin{array}{c}
\mathbf{z}^{(1)}_{}[k+1]\\
\mathbf{z}^{(2)}_{}[k+1]\\
\mathbf{z}^{(2)}_{\mathcal{U}}[k+1]\\
\end{array}\right]}_{\mathbf{z}_2[k+1]} &= 
\underbrace{\left[
\begin{array}{c|cc}
\mathbf{A}_{11}  & \multicolumn{2}{c}{\mathbf{0}} \\
\hline
\multirow{2}{*}{$\bar{\mathbf{T}}^{-1}_2 \mathbf{A}_{1X}$}
 & 
\mathbf{A}_{22} & \mathbf{0} \\
&
\star & \mathbf{A}_{2\mathcal{U}}\\
\end{array}
\right]}_{\bar{\mathbf{A}}_2=\mathbf{T}^{-1}_2\bar{\mathbf{A}}_1\mathbf{T}_2}\underbrace{\left[\begin{array}{c}
\mathbf{z}^{(1)}_{}[k]\\
\mathbf{z}^{(2)}_{}[k]\\
\mathbf{z}^{(2)}_{\mathcal{U}}[k]\\
\end{array}\right]}_{\mathbf{z}_2[k]}, \\
\mathbf{y}_2[k] &= \hspace{6.5mm} \underbrace{\left[
\begin{array}{c|cc}
\mathbf{C}_{21} & \mathbf{C}_{22} & \hspace {2mm} \mathbf{0} \end{array}\right]}_{\bar{\mathbf{C}}_2 = \mathbf{C}_2\mathbf{T}_1\mathbf{T}_2}\mathbf{z}_2[k],
\end{split}
\label{eqn:obs2}
\end{equation}
where 
\begin{equation}
\begin{split}
\bar{\mathbf{T}}^{-1}_2\mathbf{A}_{1\mathcal{U}}\bar{\mathbf{T}}_2 &= \begin{bmatrix} \mathbf{A}_{22} & \mathbf{0} \\
\star & \mathbf{A}_{2\mathcal{U}}\end{bmatrix}, \\
\mathbf{C}_{21\mathcal{U}}\bar{\mathbf{T}}_2 &= \left[\begin{array}{cc}
\mathbf{C}_{22} & \mathbf{0} \end{array}\right].
\end{split}
\end{equation}
Let $\mathbf{z}^{(2)}_{}[k] \in {\mathbb{R}}^{o_2}$. Let $\mathbf{A}_{21}$ be the matrix formed by the first $o_2$ rows of $\bar{\mathbf{T}}^{-1}_2 \mathbf{A}_{1X}$. From (\ref{eqn:obs2}), we have
\begin{equation}
\begin{split}
\mathbf{z}^{(2)}_{}[k+1]&=\mathbf{A}_{22}\mathbf{z}^{(2)}_{}[k]+\mathbf{A}_{21}\mathbf{z}^{(1)}_{}[k],\\
\mathbf{y}_2[k]&=\mathbf{C}_{22}\mathbf{z}^{(2)}_{}[k]+\mathbf{C}_{21}\mathbf{z}^{(1)}_{}[k].
\end{split}
\label{eqn:truedyn2}
\end{equation}
Following the same design procedure, we continue the sequence of transformations, one for each sensor, until we reach the $N$-th sensor.

\textbf{Step $N$} : \textit{Transformation at Sensor $N$}

Let $\bar{\mathbf{T}}_{N}$ be a non-singular transformation matrix that performs an Observable Canonical Decomposition of the pair $(\mathbf{A}_{(N-1)\mathcal{U}},  \mathbf{C}_{N(N-1)\mathcal{U}})$. Next, consider the coordinate transformation $\mathbf{z}_{N-1}[k]=\mathbf{T}_{N}\mathbf{z}_{N}[k]$, where the non-singular transformation matrix $\mathbf{T}_{N}$ is defined as follows:
\begin{equation}
\mathbf{T}_{N}=
\left[
\begin{array}{c|c|cc}
\mathbf{I}_{o_1}  & \multicolumn{3}{c}{\mathbf{0}} \\
\hline
\multirow{4}{*}{$\mathbf{0}$}
 & 
\mathbf{I}_{o_2} & \multicolumn{2}{c}{\mathbf{0}} \\
\cline{2-4}
&
\multirow{3}{*}{$\mathbf{0}$} & \hspace{-5mm} \ddots  &\vdots\\
&  &  \mathbf{I}_{o_{(N-1)}} & \mathbf{0}\\
&  & \mathbf{0} & \bar{\mathbf{T}}_{N}\\
\end{array}
\right].
\label{eqn:Tni}
\end{equation}
Using this transformation matrix, it is easy to identify that the resulting dynamics are governed by the following equations:
\begin{equation}
\begin{split}
\mathbf{z}_{N}[k+1]&=\bar{\mathbf{A}}_{N}\mathbf{z}_{N}[k],\\
\mathbf{y}_{N}[k]&=\bar{\mathbf{C}}_{N}\mathbf{z}_{N}[k],
\end{split}
\end{equation}
where 
$\bar{\mathbf{A}}_{N}$ is equal to $\bar{\mathbf{A}}$ in equation (\ref{eqn:gen_form}) and $\bar{\mathbf{C}}_{N}$ attains the  form: 
\begin{equation}
\bar{\mathbf{C}}_{N} = \left[ \begin{array}{ccccc} 
\mathbf{C}_{{N1}} & \multicolumn{1}{c}{\mathbf{C}_{{N2}}}  &  \cdots  \mathbf{C}_{N(N-1)} & \mathbf{C}_{NN}  & \multicolumn{1}{|c}{\mathbf{0}}\\
 \end{array}
 \right].
\end{equation}
Thus, by defining the similarity transformation matrix $\mathcal{T}=\prod_{i=1}^{n}\mathbf{T}_i$, we obtain the desired result.
\end{proof}

\section{Proof of Theorem \ref{theo:timevar2}}
\label{sec:proof_timevarconstr}
\begin{proof}
Following the proof technique of Theorem \ref{GEPthm}, we induct on the sub-state number and use the same notation as in the former proof. Accordingly, note that the dynamics of the composite estimation error vector for the first sub-state, namely $\bar{\mathbf{E}}^{(1)}_{}[k]$,\footnote{Note that the entries of $\bar{\mathbf{E}}^{(1)}_{}[k]$ match a topological ordering consistent with a spanning DAG rooted at node $1$ in the baseline graph $\mathcal{G}$.} is governed by the following switched linear system model: $\bar{\mathbf{E}}^{(1)}_{}[k+1]=\mathbf{M}_1[k]\bar{\mathbf{E}}^{(1)}_{}[k].$
Here, $\mathbf{M}_1[k]$ is a time-varying matrix induced by the class of switching signals $\Omega$ and is of the structure given by \eqref{eqn:bigerrmatrix}. Since $\Omega$ satisfies Assumption \ref{assump:switch}, each non-source node $i\in\mathcal{V}\setminus\{1\}$ is guaranteed to receive information from at least one of its parents in $\mathcal{P}^{(1)}_i$ in at least one switching mode over every time interval of the form $[kT,(k+1)T]$, where $k\in\mathbb{N}$. Based on our  estimation scheme, for that corresponding switching mode, the block diagonal entry corresponding to node $i$ in the matrix $\mathbf{M}_1[k]$ will be zero. With this observation in mind, consider the following dynamics: $\bar{\mathbf{E}}^{(1)}_{}[(k+1)T]=\bar{\mathbf{M}}_1(k)\bar{\mathbf{E}}^{(1)}_{}[kT],
$
where $\bar{\mathbf{M}}_1(k)=\mathbf{M}_1[(k+1)T-1]\cdots\mathbf{M}_1[kT+1]\mathbf{M}_1[kT]$. From our prior discussion, it easily follows that $\bar{\mathbf{M}}_1(k)$ is a lower block triangular matrix with zeroes on the block-diagonal corresponding to the non-source nodes in $\mathcal{V}\setminus\{1\}$ and the entry $\mathbf{(A_{11}-L_1C_{11})}^T$ corresponding to node $1$. As the pair $(\mathbf{A}_{11},\mathbf{C}_{11})$ is observable by construction, it follows using standard arguments that  $\bar{\mathbf{M}}_1(k)$ is always a Schur stable matrix.    Since $\bar{\mathbf{M}}_1(k)$ belongs to a finite set of matrices (owing to a finite number of switching modes), we can directly use \cite[Proposition 2.9]{liberzon} to establish that $\lim_{k\to\infty} \mathbf{E}^{(1)}_{}[kT] = \mathbf{0}$ and hence $\lim_{k\to\infty} \mathbf{E}^{(1)}_{}[k] = \mathbf{0}$.

Next, suppose that $\mathbf{E}^{(j)}_{}[k] $ converges to zero asymptotically $\forall j \in \{1, \cdots ,p-1\}$, where $1 \leq p-1 \leq N-1$. The composite estimation error dynamics for sub-state $p$ over an interval of length $T$ is given by
\begin{equation}
\bar{\mathbf{E}}^{(p)}_{}[(k+1)T]=\bar{\mathbf{M}}_p(k)\bar{\mathbf{E}}^{(p)}_{}[kT]+\bar{\mathbf{F}}_p(k)\bar{\mathbf{v}}^{(p)}_{},
\end{equation}
where $\bar{\mathbf{M}}_p(k)$ is defined in the same way as $\bar{\mathbf{M}}_1(k)$, and 
\begin{equation}
\bar{\mathbf{v}}^{(p)}=\begin{bmatrix} \mathbf{v}^{(p)}[kT]\\ \vdots\\ \mathbf{v}^{(p)}[(k+1)T-1]\end{bmatrix}, \mathbf{v}^{(p)}[k]= \sum_{l=1}^{p-1} \mathbf{H}_{pl}\bar{\mathbf{E}}^{(pl)}_{}[k], \nonumber
\end{equation}

\begin{equation}
\bar{\mathbf{F}}_p(k)= \begin{bmatrix} (\mathbf{M}_p[(k+1)T-1]\cdots\mathbf{M}_p[kT+1])&\cdots&\mathbf{M}_p[(k+1)T-1]&\mathbf{I}_{No_p}\end{bmatrix}.\nonumber
\end{equation}
It follows from our induction hypothesis that $\lim_{k\to\infty} \mathbf{v}^{(p)}[k]=\mathbf{0}$. Since $\bar{\mathbf{M}}_p(k)$ is Schur stable for the class of switching signals satisfying Assumption \ref{assump:switch} (in the same way as $\bar{\mathbf{M}}_1(k)$ is Schur stable), it follows from ISS and \cite[Proposition 2.9]{liberzon} that $\lim_{k\to\infty} \mathbf{E}^{(p)}_{}[k] = \mathbf{0}$. Since the update rule \eqref{eqn:unobsestimate} for the unobservable component of the state is unaffected by changes in the network structure, the rest of the proof proceeds similarly as the proof of Theorem \ref{GEPthm}.
\end{proof}

\end{document}